\documentclass{lmcs}

\usepackage{tikz}
\usepackage{hyperref}
\usetikzlibrary{automata,calc}

\newcommand{\NN}{\mathbb{N}}
\newcommand{\hdcw}{HD-tCBW}

\begin{document}

\title[How Concise are Chains of co-Büchi Automata?]{How Concise are Chains of co-Büchi Automata?}
\titlecomment{{\lsuper*}
This paper is the extended version of a paper with the same title published at GandALF 2025 \cite{gandalf2025}.
}
\thanks{Funded by Volkswagen Foundation within its \emph{Momentum} framework under project no.~9C283}	%

\author[R.~Ehlers]{Rüdiger Ehlers\lmcsorcid{0000-0002-8315-1431}}

\address{Clausthal University of Technology, Institute for Software and Systems Engineering, Arnold-Sommerfeld-Straße 1, 38678 Clausthal-Zellerfeld, Germany}	%

\begin{abstract}
{
Chains of co-Büchi automata (COCOA) have recently been introduced as a new canonical model for representing arbitrary $\omega$-regular languages. They can be minimized in polynomial time and are hence an attractive language representation for applications in which normally, deterministic $\omega$-automata are used.
While it is known how to build COCOA from deterministic parity automata, little is currently known about their relationship to automaton models introduced earlier than COCOA.

In this paper, we 
analyze the conciseness of chains of co-Büchi automata. We provide three main results and give an overview of the implications of these results.
First of all, we
show that even in the case that all automata in the chain are deterministic, chains of co-Büchi automata can be exponentially more concise than deterministic parity automata. %
We then present two main results that together negatively answer the question if this conciseness is retained when performing Boolean operations (such as disjunction, conjunction, and complementation) over COCOA. For the binary operations, we show that there exist families of languages for which their application leads to an exponential growth of the sizes of the automata. The families have the property that when representing them using deterministic parity automata, taking the disjunction or conjunction of the family elements only requires a polynomial blow-up.
We finally show that an exponential blow-up is also unavoidable when complementing a COCOA, as this operation can require redistributing with which colors words need to be recognized.
}

\end{abstract}

\maketitle

\section{Introduction}

Automata over infinite words are a classical model for representing the specification of a reactive system.
They augment temporal logics such as linear temporal logic (LTL, \cite{DBLP:conf/focs/Pnueli77}) and linear dynamic logic (LDL, \cite{DBLP:conf/ijcai/GiacomoV13,DBLP:journals/iandc/FaymonvilleZ17}) by providing an intermediate representation for a specification that is structured in a way so that it can be used directly in verification and synthesis algorithms.
While for classical model checking of finite-state systems, non-deterministic automata with a Büchi acceptance condition suffice, for some applications, such as reactive synthesis and probabilistic model checking, richer automaton types are employed. 
In this context, deterministic automata with parity acceptance are particularly interesting as when a specification is given as such, the reactive synthesis problem over the specification can be reduced to solving a parity game based on the state space structure of the automaton \cite{DBLP:reference/mc/BloemCJ18}.

Unfortunately, translation procedures from temporal logic to
deterministic parity automata cannot avoid a substantial blow-up in the worst case, which complicates employing them in reactive synthesis. For instance, for a specification in LTL, a doubly-exponential worst-case blow-up is known \cite{DBLP:conf/mochart/KupfermanR10}. However, even for languages that do not require such huge automata, current translation procedures for obtaining deterministic parity automata can compute unnecessarily large automata, caused by them only applying heuristics for size minimization, as deterministic parity automaton minimization is NP-hard \cite{DBLP:conf/fsttcs/Schewe10,DBLP:journals/corr/abs-2504-20553}, 

To counter this problem, chains of co-Büchi automata (COCOA) have recently been proposed as a new model for $\omega$-regular languages \cite{DBLP:conf/fsttcs/EhlersS22}.
In a COCOA, the language to be represented is split into a falling chain of co-Büchi languages, where each of the co-Büchi languages is represented as a \emph{history-deterministic} co-Büchi automaton with \emph{transition-based acceptance} (\hdcw). In this context, tran\-si\-tion-based acceptance refers to the transitions being accepting or rejecting rather than the states. This particular type of co-Büchi automata is minimizable in polynomial time \cite{DBLP:journals/lmcs/RadiK22}, so that each automaton in the chain can be minimized separately. To employ these automata in a canonical and polynomial-time minimizable model for arbitrary $\omega$-regular languages, a canonical split of an $\omega$-regular language to co-Büchi automata was defined \cite{DBLP:conf/fsttcs/EhlersS22}. COCOA can be thought of as assigning a \emph{color} to each word, just as deterministic parity automata do. A word then has a color of $i$ (for some $i \in \NN$) if the $i$th automaton in the chain accepts the word, but no automaton later in the chain accepts the word. The core contribution of the COCOA definition is a concretization of which color should be assigned to each word, and this concretization is not based on some automaton representation of the language, but only on the language itself, called the \emph{natural color} of the respective word.

COCOA have already found first use in reactive synthesis \cite{DBLP:conf/tacas/EhlersK24}, based on a procedure for translating deterministic parity automata to COCOA \cite{DBLP:conf/fsttcs/EhlersS22}. Given that polynomial-time minimization is an attractive property for future applications as well, it makes sense to have a closer look at the properties of COCOA and their relationship to earlier automata types, in particular in relation to deterministic parity automata, which they have the potential of replacing in some applications. 

In particular, we need to understand their conciseness in relation to previous automaton models, as doing so provides an indication of in which cases COCOA are a useful replacement in verification and synthesis procedures that currently base on other automaton types. Furthermore, with polynomial-time minimization of COCOA available, future applications of this representation for $\omega$-regular languages may build specification models in a compositional way while minimizing the intermediate automata after each step. To understand whether such an approach can be feasible, we need to understand the complexity of performing the usual Boolean operations on automata, most importantly conjunction (language intersection) and complementation.

In this paper, we provide a study of the conciseness of COCOA. We compare them against deterministic parity automata and analyze how performing Boolean operations on languages represented by COCOA affects their conciseness. Apart from summarizing how some existing results on deterministic co-Büchi automata transfer to the COCOA case, we provide three new COCOA-specific technical results:
\begin{enumerate}
\item 
We show that COCOA can be exponentially more concise than deterministic parity automata (DPW) even when the co-Büchi languages in the COCOA are representable as small deterministic co-Büchi automata and when the overall language only has one residual language.
{While it was previously known that COCOA can be exponentially more concise than deterministic parity automata, this was due to the automata in the COCOA being history-deterministic, and \hdcw{} are known to be exponentially more concise than deterministic automata (for some languages). The new result in this paper shows that COCOA can be exponentially more concise than DPW even when not making use of history-determinism for the chain elements.} Our results also imply that these sources of conciseness to not stack.

\item We show that exponential conciseness of COCOA over deterministic parity automata can be lost when performing binary Boolean operations (such as conjunction or disjunction) on COCOA. In particular, such Boolean operations can require an exponential growth in the number of states even in cases in which for deterministic parity automata, such a growth is not necessary.
\item Finally, we prove that performing complementation of COCOA, as a Boolean operation with only a single operand, also leads to an exponential blow-up in the worst case. This is caused by words with the same natural color in the uncomplemented COCOA potentially having different natural colors in the complemented COCOA. Distinguishing between such words can require distinguishing between exponentially many residual languages of the original COCOA in a single chain element of the complement COCOA, which shatters conciseness. 
\end{enumerate}
All three results shed light on the fundamental properties of COCOA. In the first case, the example family of languages defined for the result shows that even with small automata in a chain, complex liveness languages can be composed.
This even holds when the automata do not make use of the additional conciseness of history-deterministic co-Büchi automata.
The second technical result shows that the property of a COCOA to have a number of residual languages that is exponential in their size can be lost when applying binary Boolean operations. It hence demonstrates that future procedures for performing Boolean operations on COCOA will need to have an exponential lower bound on the sizes of the resulting COCOA.
Finally, the third technical result demonstrates that for some operations on COCOA, the residual languages of co-Büchi automata that appear late in a chain can have an effect on on the first co-Büchi automaton of an output COCOA.

This paper is structured as follows: After stating some preliminaries, we give a summary of the ideas behind COCOA in Section~\ref{sec:COCOA}. Section~\ref{sec:conciseness} summarizes the implications of known results on the conciseness of COCOA and provides the first new technical result. Section~\ref{sec:conjunction} then contains the lower bound on the blow-up incurred by binary Boolean operations on COCOA. 
The third technical result concerning the complementation of COCOA is given in Section~\ref{sec:complementation}.
The paper closes with a discussion of the obtained results in Section~\ref{sec:conclusion}.

{This paper is the extended version of a paper with the same name published at the GandALF 2025 conference held in September 2025 in Valletta, Malta. It extends the conference version by the results in Section~\ref{sec:complementation} (on the exponential blow-up when complementing COCOA) as well as by proofs for Lemma~\ref{lem:correctNaturalColors}, Theorem~\ref{thm:conjunctionBecomesBig}, and Lemma~\ref{lem:conjunctionDisjunctionOfDCW}. The proof 
of Theorem~\ref{thm:conjunctionBecomesBig} includes the newly appearing observations~\ref{obs:three} to \ref{obs:five}.
Finally, some details were added to the proof of Proposition~\ref{proposition:paritySizes}. %
}

\section{Preliminaries}

\textbf{Languages:} For a finite set $\Sigma$ as \emph{alphabet}, let $\Sigma^*$ denote the set of finite words over $\Sigma$, and $\Sigma^\omega$ be the set of infinite words over $\Sigma$. A subset $L \subseteq \Sigma^\omega$ is also called a \emph{language}. Given a language $L$ and some finite word $w \in \Sigma^*$, we say that $L|_w = \{w' \in \Sigma^\omega \mid ww' \in L\}$ is the \emph{residual language} of $L$ over $w$. Given a word $w = w_0 w_1 \ldots \in \Sigma^\omega$ and some language $L$, we say that an infinite word $w' = w_0 w_1 \ldots w_i \tilde w w_{i+1} w_{i+2} \ldots$ results from a \emph{residual language invariant injection} of $\tilde w$ at position $i \in \NN$ if $L|_{w_0 \ldots w_i} = L|_{w_0 \ldots w_i \tilde w}$.

\textbf{Automata:} Some languages, in particular the \emph{$\omega$-regular languages}, can be represented by \emph{parity automata}. We only consider automata with \emph{transition-based acceptance} in this paper. These are tuples of the form $\mathcal{A} = (Q,\Sigma,\delta,q_0)$ in which $Q$ is a finite set of states, $\Sigma$ is the alphabet, $q_0 \in Q$ is the initial state of the automaton, and $\delta \subseteq Q \times \Sigma \times Q \times \NN$ is its \emph{transition relation}.

Given a word $w = w_0 w_1 \ldots \in \Sigma^\omega$, we say that $w$ induces an infinite run $\pi = \pi_0 \pi_1 \ldots \in Q^\omega$ together with a sequence of \emph{colors} $\rho = \rho_0 \rho_1 \ldots \in \NN^\omega$ if we have $\pi_0 = q_0$ and for all $i \in \NN$, we have $(\pi_i,w_i,\pi_{i+1},\rho_i) \in \delta$. In this paper, we only consider automata that are \emph{input-complete}, i.e., for which for each state/letter combination $(q,x)$, there exists at least one pair $(q',c)$ with $(q,x,q',c) \in \delta$. We furthermore only consider automata for which the color $c$ does not depend on the transition taken, so that for each $(q,x)$, there is only one value $c$ with $(q,x,q',c) \in \delta$ for some $q'$.

A run is accepting if for the corresponding color sequence $\rho$ (which is unique), we have that the lowest color occurring infinitely often in it is even. This color is also called the \emph{dominating color} of the run. A word is accepted by $\mathcal{A}$ if there exists an accepting run for it. The language of $\mathcal{A}$, written as $\mathcal{L}(\mathcal{A})$, is the set of words with accepting runs.
An automaton is said to be deterministic if for every $(q,x) \in Q \times \Sigma$, there exists exactly one combination $(q',c) \in Q \times \NN$ with $(q,x,q',c) \in \delta$. In such a case, we also refer to the dominating color of the unique run as the color with which the automaton \emph{recognizes} the word. We say that $\mathcal{A}$ is a \emph{co-Büchi automaton} if the only colors occurring along transitions in $\mathcal{A}$ are 1 and 2. A \emph{co-Büchi language} is a language of some co-Büchi automaton. 

We say that an automaton $\mathcal{A}$ represents the \emph{disjunction} of some automata $\mathcal{A}_1$ and $\mathcal{A}_2$ if $\mathcal{L}(\mathcal{A}) = \mathcal{L}(\mathcal{A}_1) \cup \mathcal{L}(\mathcal{A}_2)$. It represents the \emph{conjunction} of  $\mathcal{A}_1$ and $\mathcal{A}_2$ if $\mathcal{L}(\mathcal{A}) = \mathcal{L}(\mathcal{A}_1) \cap \mathcal{L}(\mathcal{A}_2)$. {The \emph{size} of an automaton is defined to be the number of its states.}

\textbf{History-deterministic automata:} Parity automata, as defined above, are not necessarily deterministic. We consider \emph{history-deterministic} co-Büchi automata (\hdcw{}) in particular. For them, there exists some \emph{advice} function $f : \Sigma^* \rightarrow Q$ such that for each word, if and only if $w = w_0 w_1 \ldots \in \mathcal{L}(\mathcal{A})$, the sequence $q_0 f(w_0) f(w_1 w_2 \ldots) \ldots$ is a valid accepting run of the automaton. Abu Radi and Kupferman~\cite{DBLP:journals/lmcs/RadiK22} showed how to minimize such automata, and in minimized automata, for every state/letter combination, all transitions have the same color (so that the assumption from above is justified). History-deterministic co-Büchi automata are also called \emph{good-for-games} co-Büchi automata in the literature. The sets of languages representable by \hdcw{} and deterministic co-Büchi automata are the same. Deterministic parity automata (DPW) are however strictly more expressive.
We also sometimes represent automata in a graphical notation, where states are circles, transitions are arrows between circles, and the initial state is marked by an arrow from a dot. In co-Büchi automata, dashed arrows represent \emph{rejecting transitions} (with color $1$), while the solid arrows represent \emph{accepting transitions} (with color $2$). For parity automata, the edges are labeled by the color numbers in addition to their alphabet letters.

\textbf{SCCs:} Given an automaton $\mathcal{A} = (Q,\Sigma,\delta,q_0)$, we say that some tuple $(Q',\delta')$ with $Q' \subseteq Q$ and $\delta' \subseteq \delta$ is a \emph{strongly connected component} (SCC) of $\mathcal{A}$ if for each $q, q' \in Q'$, there exists a sequence of transitions within $\delta'$ for reaching $q'$ from $q$. Similarly, every transition within $\delta'$ is used in some such sequence.

{\textbf{Temporal logic and $\omega$-regular expressions:}
\emph{Linear temporal logic} (LTL, \cite{DBLP:conf/focs/Pnueli77}) is a formalism for expressing (some) languages over $\Sigma = 2^\mathsf{AP}$ for a set $\mathsf{AP}$. It is known that LTL can be translated to deterministic parity automata of size doubly-exponential in the sizes of the LTL properties, and this blow-up bound is tight (see, e.g., \cite{DBLP:conf/tacas/EsparzaKRS17}).
We also use \emph{$\omega$-regular expressions} for stating some languages. These extend classical regular expressions by a symbol for infinite repetition, namely $^\omega$.
}

\section{A short introduction to chains of co-Büchi automata}
\label{sec:COCOA}

Chains of co-Büchi automata (COCOA, used as both the singular and plural form) provide a canonical representation for arbitrary $\omega$-regular languages.
Let a language $L$ over an alphabet $\Sigma$ be given. The starting point for a COCOA representation of $L$ is the decomposition of $\Sigma^\omega$ into a chain of languages $L_1 \supset L_2 \supset \ldots \supset L_n$.
A word $w$ is in the language represented by the chain, also denoted as $\mathcal{L}(L_1, \ldots, L_n)$ henceforth, if the \emph{highest} index $i$ such that $w \in L_i$ is even or $w \notin L_1$.
Each language $L_i$ (for some $1 \leq i \leq n$) represents the set of words whose \emph{natural color} (with respect to $L$) is at least $i$. The natural color of a word is the minimal color in which a word is \emph{at home}, which in turn is defined as follows:

\begin{defi}[\cite{DBLP:conf/fsttcs/EhlersS22}, Def.~1]
\label{def:naturalColorsOfWords}
Let $L$ be a language and $i \in \NN$. We say that a word $w$ is at home in a color of $i$ if there exists a sequence of injection points $J \subset \NN$ such that for all words $w'$ that result from injecting residual language invariant words at word positions in $J$ into $w$, {we either have that
\begin{itemize}
\item $w'$ is at home in a color strictly smaller than $i$, or
\item both $w$ and $w'$ are in $L$ and $i$ is even, or both $w$ and $w'$ are not in $L$ and $i$ is odd.
\end{itemize}
Note that in the case of $i=0$, only the second case can} apply.
\end{defi}
The concept of the natural color of a word generalizes the idea of colors in a parity automaton in a way that is agnostic to the concrete choice of automaton for representing the language. The inductive definition above starts from color $0$, so that the languages at each level are uniquely defined.

With this definition, not only is a chain $L_1, \ldots, L_n$ of languages uniquely defined for each $\omega$-regular language $L$, but we also have that for each $1 \leq i \leq n$, the language $L_i$ is a co-Büchi language \cite{DBLP:conf/fsttcs/EhlersS22}, i.e., it can be encoded into a co-Büchi automaton. Hence, we can represent the chain of languages $L_1, \ldots, L_n$ by a chain of history-deterministic co-Büchi automata $A_1, \ldots, A_n$. Each of these automata can be minimized and made canonical in polynomial time \cite{DBLP:journals/lmcs/RadiK22}. 
Since the representation of $L$ as a chain of co-Büchi languages $L_1, \ldots, L_n$ is also canonical, we %
obtain a canonical representation of $L$. 

Details on the COCOA language representation can be found in the paper introducing COCOA \cite{DBLP:conf/fsttcs/EhlersS22} and in a video recording of a presentation of the paper's concepts with  additional examples
\cite{naturalColorsVideo}. 

\subsection{An example COCOA}

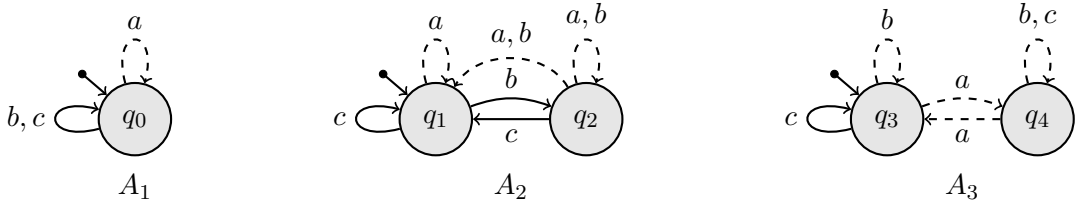
\begin{figure}
\centering\begin{tikzpicture}

\begin{scope}[xshift=-1cm]
\node[state,thick,fill=black!10!white] (a) at (0,0) {$q_0$};
\draw[thick,->,dashed] (a) to[loop above] node[above] {$a$} (a);
\draw[thick,->] (a) to[loop left] node[left] {$b,c$} (a);
\draw[fill=black] ($(a)+(-0.7,0.6)$) circle (0.05cm);
\draw[->,thick] ($(a)+(-0.7,0.6)$) -- (a);

\end{scope}

\node[state,thick,fill=black!10!white] (b) at (3,0) {$q_1$};
\node[state,thick,fill=black!10!white] (c) at (5,0) {$q_2$};
\draw[thick,->,dashed] (b) to[loop above] node[above] {$a$} (b);
\draw[thick,->,dashed] (c) to[loop above] node[above] {$a,b$} (c);
\draw[thick,->,dashed] (c) to[bend right=60] node[above] {$a,b$} (b);
\draw[thick,->] (c) to[bend left=0] node[below] {$c$} (b);
\draw[thick,->] (b) to[bend left=20] node[above] {$b$} (c);
\draw[thick,->] (b) to[loop left] node[left] {$c$} (b);
\draw[fill=black] ($(b)+(-0.7,0.6)$) circle (0.05cm);
\draw[->,thick] ($(b)+(-0.7,0.6)$) -- (b);

\begin{scope}[xshift=1cm]
\node[state,thick,fill=black!10!white] (d) at (8,0) {$q_3$};
\node[state,thick,fill=black!10!white] (e) at (10,0) {$q_4$};
\draw[thick,->,dashed] (d) to[loop above] node[above] {$b$} (d);
\draw[thick,->,dashed] (e) to[loop above] node[above] {$b,c$} (e);
\draw[thick,->] (d) to[loop left] node[left] {$c$} (d);
\draw[thick,->,dashed] (e) to[bend left=0] node[below] {$a$} (d);
\draw[thick,->,dashed] (d) to[bend left=20] node[above] {$a$} (e);
\end{scope}

\draw[fill=black] ($(d)+(-0.7,0.6)$) circle (0.05cm);
\draw[->,thick] ($(d)+(-0.7,0.6)$) -- (d);

\node at ($(a)+(0,-0.9)$) {$A_1$};
\node at ($0.5*(b)+0.5*(c)+(0,-0.9)$) {$A_2$};
\node at ($0.5*(d)+0.5*(e)+(0,-0.9)$) {$A_3$};

\end{tikzpicture}
\caption{An example COCOA}
\label{fig:exampleCOCOA}
\end{figure}

Figure~\ref{fig:exampleCOCOA} shows an example COCOA consisting of three automata, all over the alphabet $\Sigma = \{a,b,c\}$, together representing some language $L$. The chain's language contains the words with an infinite number of $a$ letters (with a natural color of $0$) as well as words that satisfy three conditions:
\begin{itemize}
\item the word ultimately only consists of $b$s and $c$s,
\item eventually, every $b$ is immediately followed by a $c$, and
\item if there is a finite even number of $a$ letters in the word, there are infinitely many $b$s.
\end{itemize}
Words satisfying these three conditions but only having finitely many $a$s have a natural color of $2$. Words that are not in $L$ have natural colors of $1$ or $3$. The COCOA hence recognizes words with four different natural colors, and it follows from the existing translation procedure from deterministic parity automata to COCOA \cite{DBLP:conf/fsttcs/EhlersS22} that every deterministic parity automaton for this language also needs at least four colors. 
The colors represent how often by injecting residual language invariant finite words an infinite number of times, words can alternate between being in $L$ or not.
In this example, the word $c^\omega$ has color $3$ and is hence not in $L$. By injecting $b$s such that the resulting word never has two $b$ letters in a row, the word becomes contained in $L$. By subsequently injecting $bb$ infinitely often, the word  leaves $L$ again. Finally, by then injecting $a$ infinitely often, the final word is rejected by $A_1$ and hence in $L$. The overall language $L$ represented by the COCOA has two residual languages, but only $A_3$ tracks them and not $A_2$ or $A_1$. The relevance of the injection point set $J$ in Definition~\ref{def:naturalColorsOfWords} is not exemplified in the COCOA in Figure~\ref{fig:exampleCOCOA}, as for all concretely given COCOA discussed in the following sections, this set can be freely chosen and is hence not of relevance.

\subsection{Some additional definitions and notes in the context of COCOA}
\label{subsec:additionalCOCOALemmata}

For convenience, whenever we are dealing with a COCOA $A_1, \ldots, A_n$ in the following, we will assume that $A_0 = \Sigma^\omega$ and $A_{n+1} = \emptyset$, as this avoids dealing with special cases in some constructions while not affecting the definition of the COCOA's language.
To avoid cluttering the exposition in the following, the word \emph{injection} always refers to a residual language invariant word injection. When a word $w'$ is the result of a residual language invariant word injection into some word $w$, we say that $w'$ \emph{extends} $w$.
We define the sum of the automaton sizes in a COCOA to be the size of the COCOA.

The condition for a chain of co-Büchi automata  $A_1, \ldots, A_n$ to represent a language $L$ can be equivalently stated as requiring $A_1$ to reject the words with a natural color of $0$ (with respect to $L$) and that for each $1 \leq i \leq n$, the words accepted by $A_i$ but rejected by $A_{i+1}$ (if $i<n$) are the ones with a natural color of $i$ (with respect to $L$).
The definitions above also imply that the natural language of a word $w$ can only decrease by injecting letters into $w$ {(if the set $J$ of positions to inject at is chosen according to the requirements of Definition~\ref{def:naturalColorsOfWords})}.

\section{On the conciseness of COCOA}
\label{sec:conciseness}

In this section, we will relate the sizes of deterministic parity automata to the sizes of COCOA (for the same languages).
We summarize the implications of existing results on the conciseness of COCOA and augment them by new insights.

\paragraph{DPW conciseness over COCOA} For starters, the translation by Ehlers and Schewe \cite{DBLP:conf/fsttcs/EhlersS22} for obtaining a COCOA from a deterministic parity automaton with $n$ states and $c$ colors yields COCOA with at most $c$ history-deterministic co-Büchi automata, each having at most $n$ many states. Even more, since the construction by Ehlers and Schewe minimizes the numbers of colors on-the-fly, this fact also holds for $c$ being the minimal number of colors that \emph{any} DPW for the language has. Hence, a COCOA can only be polynomially larger than a deterministic parity automaton for the same language, and the factor by which it can be larger is bounded by the number of colors.
This bound is also tight:
\begin{prop}[Appears to not have been stated previously elsewhere]
Let $\inf$ be the function mapping a sequence to the set of elements occurring infinitely often in the sequence.
For every $k \in \NN$, the language $L^k = \{w \in \{1, \ldots, k\}^\omega \mid \min(\inf(w)) \text{ is even}\}$ can be represented by a deterministic parity automaton with a single state and $k$ colors, but every COCOA for the same language needs at least $k$ levels (with one state on each level). 
\end{prop}
\begin{proof}
A deterministic parity automaton with a single state can be built with self-loops for all letters that use the letter as the respective color. The existing procedure for translating a DPW to a COCOA \cite{DBLP:conf/fsttcs/EhlersS22} then builds a COCOA $A_1, \ldots, A_k$ for this language in which on each level $i$, the words ending with $(\{i, \ldots, k\})^\omega$ are accepted. Overall, we have a blow-up by a factor of $k$, while $k$ is the number of colors in the deterministic parity automaton that we start with. 
\end{proof}

\paragraph{LTL $\rightarrow$ COCOA}
Before discussing that COCOA can also be more concise than DPW, we look at an area in which they have the same conciseness. In particular, a translation from LTL to automata has the same worst-case blow-up lower bound for DPW and COCOA, namely doubly-exponential.

This follows from an existing proof of the doubly-exponential lower bound for translating from LTL to deterministic Büchi automata (whenever possible). Kupferman and Rosenberg gave multiple versions of such proofs for the cases of fixed and non-fixed alphabets \cite{DBLP:conf/mochart/KupfermanR10}. All proofs have in common that a family of languages is built that has a doubly-exponential number of residual languages (in the sizes of the LTL formula). This can be seen from the fact that their languages only contain words that end with $\#^\omega$ for some character $\#$ in the alphabet, and hence only a doubly-exponential blow-up in the number of residual languages can cause the automata to be so big.

The complements of these languages are representable by co-Büchi automata. Furthermore, minimal \hdcw{} are \emph{semantically deterministic}, meaning that for each state $q$ in the automaton $A = (Q,\Sigma,\delta,q_0)$ reachable under a prefix word $\tilde w$, we have $\mathcal{L}((Q,\Sigma,\delta,q)) = \{w \in \Sigma^\omega \mid \tilde w w \in \mathcal{L}(A)\}$. If there is a doubly-exponential number of residual languages in $A$, we have that $A$ then needs at least a doubly-exponential number of states.
As a consequence, COCOA for these languages also need to be of doubly-exponential size, as a COCOA for a co-Büchi language consists of only a single \hdcw{} for the language.

\paragraph*{COCOA conciseness over DPW}
Let us now identify if and how COCOA can be more concise than DPWs.
For starters, it was shown that \hdcw{} can be exponentially more concise than deterministic co-Büchi automata \cite{DBLP:conf/icalp/KuperbergS15}. 
Since parity automata are co-Büchi type \cite{DBLP:journals/ijfcs/KupfermanMM06}, deterministic co-Büchi word automata cannot be less concise than deterministic parity automata.
Since furthermore COCOA for co-Büchi languages consist of a single history-deterministic co-Büchi automaton, we overall obtain that COCOA can be exponentially more concise than DPW.

We can also employ some existing results for showing that COCOA cannot be doubly-exponentially more concise than deterministic parity automata:

\begin{prop}[Already appearing in abbreviated form in \cite{DBLP:journals/corr/abs-2410-01021} based on remarks in \cite{DBLP:conf/tacas/EhlersK24}]
\label{proposition:COCOAtoDPWExponentialIsEnough}
Let $({A}_1, \ldots, {A}_k)$ be a COCOA. There exists a deterministic parity automaton for the same language that has a number of states that is exponential in $|{A}_1|+\ldots+|{A}_k|$.
\end{prop}
\begin{proof}
Translating a non-deterministic co-Büchi automaton to a deterministic co-Büchi automaton can be performed with an exponential blow-up \cite{DBLP:conf/icalp/BokerKR10} using the Miyano-Hayashi construction \cite{MIYANO1984321}. Doing so for each automaton in the COCOA yields a sequence of deterministic co-Büchi automata $\mathcal{D}_1, \ldots, \mathcal{D}_k$, where for each $1 \leq j \leq k$, we have $|\mathcal{D}_j| \leq 3^{|A_j|}$.

Let for each $1 \leq j \leq k$ be $\mathcal{D}_j = (Q^j,\Sigma,\delta^j,q^j_0)$. We can construct a deterministic parity automaton $\mathcal{P} = (Q^P,\Sigma,\delta^P,q^P_0)$ for the language of the COCOA as follows (using a construction from \cite{DBLP:journals/corr/abs-2410-01021}):
\allowdisplaybreaks
\begin{align*}
Q^P & = Q^1 \times \ldots \times Q^k \\
\delta^P((q^1, \ldots, q^k),x) & = ((q'^1, \ldots, q'^k),c) \text{ s.t. } \exists c^1, \ldots, c^k \in \NN. (q'^1,c^1) \in \delta^1(q^1,x), \ldots,\\
& \quad \quad (q'^k,c^k) \in \delta^k(q^k,x), c = \min(\{k \} \cup \{ j \in \{0, \ldots, k-1\} \mid c^{j+1}=2 \}) \\
q^P_0 & = (q^1_0, \ldots, q^k_0)
\end{align*}
To see that $\mathcal{P}$ has the right language, assume that for some word $w$, its natural color is $j$ for some $0 \leq j \leq k$. Then, all automata $\mathcal{D}_1, \ldots, \mathcal{D}_j$ accept the word while the automata $\mathcal{D}_{j+1}, \ldots, \mathcal{D}_k$ reject the word. Since $\mathcal{P}$ simulates all these automata in parallel, infinitely often the color $c$ along transitions in the run for $w$ will be $j$, but only finitely often the color will be in $\{0, \ldots, j-1\}$. This means that $\mathcal{P}$ accepts $w$ if and only if $j$ is even, which proves that $\mathcal{P}$ has the right language.

We have that $|\mathcal{P}| \leq |\mathcal{D}_1| \cdot \ldots \cdot |\mathcal{D}_k| \leq 3^{|{A}_1|} \cdot \ldots \cdot 3^{|{A}_k|} = 3^{|{A}_1| + \ldots + |{A}_k|} $. Overall, the blow-up of the translation is hence exponential.
\end{proof}

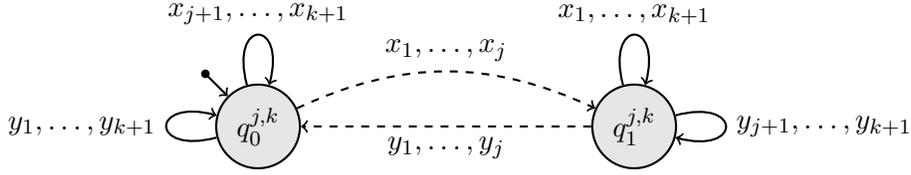
\begin{figure}
\centering\begin{tikzpicture}

\node[state,thick,fill=black!10!white] (a) at (0,0) {$q^{j,k}_0$};
\node[state,thick,fill=black!10!white] (b) at (5,0) {$q^{j,k}_1$};

\draw[thick,->] (a) to[loop above] node[above] {$x_{j+1}, \ldots, x_{k+1}$} (a);
\draw[thick,->] (a) to[loop left] node[left] {$y_{1}, \ldots, y_{k+1}$} (a);

\draw[thick,->] (b) to[loop above] node[above] {$x_{1}, \ldots, x_{k+1}$} (b);
\draw[thick,->] (b) to[loop right] node[right] {$y_{j+1}, \ldots, y_{k+1}$} (b);

\draw[thick,->,dashed] (a) to[bend left=25] node[above] {$x_{1}, \ldots, x_j$} (b);
\draw[thick,->,dashed] (b) to[bend left=0] node[below] {$y_{1}, \ldots, y_j$} (a);

\draw[fill=black] ($(a)+(-0.7,0.7)$) circle (0.05cm);
\draw[->,thick] ($(a)+(-0.7,0.7)$) -- (a);

\end{tikzpicture}
\caption{A deterministic co-Büchi automaton (parametrized for some $k \in \NN$ and $1 \leq j \leq n$) for the co-Büchi languages used in the proof of Theorem~\ref{thm:conciseNessWithSingleSuffixLanguage}}
\label{fig:firstResultSingleLevelDCW}
\end{figure}

So at a first glance, the conciseness of COCOA over DPW has been characterized to precisely singly-exponential.
What cannot be easily derived from existing results, however, is why exactly a COCOA can be exponentially more concise than a deterministic parity automaton. In particular, it may be possible that 
there are also factors other than the conciseness of history-deterministic co-Büchi automata that contribute to the  conciseness of COCOA, but they do not \emph{stack}.

It turns out that this is the case, as we show next. Even in the case that the co-Büchi languages on each level of a COCOA are representable as two-state deterministic co-Büchi automata, a parity automaton for the represented language may need exponentially more states.

\begin{thm}
\label{thm:conciseNessWithSingleSuffixLanguage}
There exists a family of COCOA $\mathcal{C}^1, \mathcal{C}^2, \ldots$ for which for each COCOA $\mathcal{C}^k = (A^k_1, \ldots, A^k_k)$, we have that for all $1 \leq j \leq k$, the history-deterministic co-Büchi automaton ${A}^k_j$ only has two states, the language of $\mathcal{C}^k$ only has a single residual language, and every deterministic parity automaton $\mathcal{P}^k$ for the language of $\mathcal{C}^k$ needs at least $2^k$ states.
\end{thm}
For the proof of this theorem, we first define a suitable family of languages. 
\begin{defi}
For every $k \in \mathbb{N}$, we define $\mathcal{C}^k = (A^k_1, \ldots, A^k_k)$ so that the co-Büchi automata in the COCOA have the joint alphabet $\Sigma^k = \{x_1, \ldots, x_{k+1}, y_1, \ldots, y_{k+1}\}$ and such that for each $1 \leq j \leq k$, a deterministic co-Büchi automaton for $A^k_j$ can be given as in Figure~\ref{fig:firstResultSingleLevelDCW}.
\end{defi}
Intuitively, every automaton $A^k_i$ in a COCOA $\mathcal{C}^k$ accepts those words in which either the letters $x_1 \ldots x_i$ appear only finitely often or the letters $y_1 \ldots y_i$ appear only finitely often. Figure~\ref{fig:dpaForC2} depicts a minimally sized DPW $\mathcal{P}^2$ for $\mathcal{L}(\mathcal{C}^2)$ and provides some intuition on why a DPW for such a COCOA may need to be large: in order to ensure that words are accepted by the DPW that are rejected by $A^2_1$, the DPW's state set needs to be split into those states corresponding to state $q_0^{1,2}$ (on the left) and those corresponding to $q_1^{1,2}$ (on the right), so that a run switching between these infinitely often is accepting. This is implemented by the transitions between the left and right parts of the the DPW having a color of $0$, which is then the dominating color of the run. Within \emph{each} of these separate state sets, however, we also need a split between states corresponding to $q_0^{2,2}$ (the bottom two states in the DPW) and those corresponding to $q_1^{2,2}$ (the top two states in the DPW) to detect when a word should be rejected by the DPW due to it not being accepted by $A^2_2$. 
Transitions between bottom and top states have a color of $1$ to implement that the word is rejected by $\mathcal{P}^2$ if the word is rejected by $A^2_2$ (but accepted by $A^2_1$).
Such a nesting of states from different co-Büchi automata in $\mathcal{C}^k$ is indeed unavoidable, as we show next in order to prove Theorem~\ref{thm:conciseNessWithSingleSuffixLanguage}.

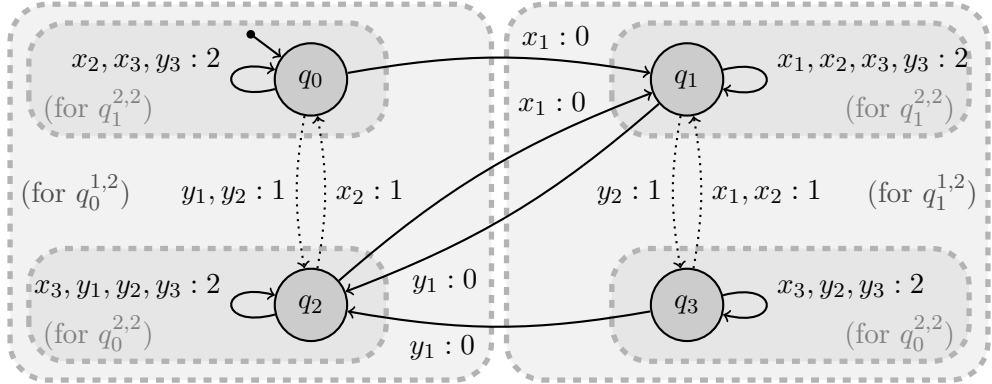
\begin{figure}
\centering\begin{tikzpicture}

\draw[line width=2pt,dashed,rounded corners=0.5cm,color=black!30!white,fill=black!5!white] (-4,-4) rectangle (2.4,1);
\draw[line width=2pt,dashed,rounded corners=0.5cm,color=black!30!white,fill=black!5!white] (2.6,-4) rectangle (9,1);

\draw[line width=2pt,dashed,rounded corners=0.5cm,color=black!30!white,fill=black!10!white] (-3.75,-3.75) rectangle +(4.75,1.5);
\draw[line width=2pt,dashed,rounded corners=0.5cm,color=black!30!white,fill=black!10!white] (-3.75,-0.75) rectangle +(4.75,1.5);
\draw[line width=2pt,dashed,rounded corners=0.5cm,color=black!30!white,fill=black!10!white] (4,-3.75) rectangle +(4.75,1.5);
\draw[line width=2pt,dashed,rounded corners=0.5cm,color=black!30!white,fill=black!10!white] (4,-0.75) rectangle +(4.75,1.5);

\node[anchor=west,color=black!70!white] at (-4,-1.5) {(for $q^{1,2}_0$)};
\node[anchor=east,color=black!70!white] at (9,-1.5) {(for $q^{1,2}_1$)};
\node[anchor=south west,color=black!50!white] at (-3.7,-3.75) {(for $q^{2,2}_0$)};
\node[anchor=south west,color=black!50!white] at (-3.7,-0.75) {(for $q^{2,2}_1$)};
\node[anchor=south east,color=black!50!white] at (8.7,-3.75) {(for $q^{2,2}_0$)};
\node[anchor=south east,color=black!50!white] at (8.7,-0.75) {(for $q^{2,2}_1$)};

\node[state,thick,fill=black!20!white] (a) at (0,0) {$q_0$};
\node[state,thick,fill=black!20!white] (b) at (5,0) {$q_1$};
\node[state,thick,fill=black!20!white] (c) at (0,-3) {$q_2$};
\node[state,thick,fill=black!20!white] (d) at (5,-3) {$q_3$};

\draw[thick,->] (a) to[loop left] node[left,yshift=0.25cm] {$x_2, x_3, y_3: 2$} (a);
\draw[thick,->] (b) to[loop right] node[right,yshift=0.25cm] {$x_1, x_2, x_3, y_3: 2$} (b);
\draw[thick,->] (c) to[loop left] node[left,yshift=0.25cm] {$x_3, y_1, y_2, y_3: 2$} (c);
\draw[thick,->] (d) to[loop right] node[right,yshift=0.25cm] {$x_3, y_2, y_3: 2$} (d); 

\draw[thick,->,pos=0.7] (a) to[bend left=10] node[above] {$x_{1}: 0$} (b);
\draw[thick,->] (c) to[bend left=10] node[above left=-1mm,pos=0.83] {$x_{1}: 0$} (b);
\draw[thick,->] (b) to[bend left=10] node[below right=-1mm,pos=0.83] {$y_{1}: 0$} (c);
\draw[thick,->,pos=0.7] (d) to[bend left=10] node[below] {$y_{1}: 0$} (c);

\draw[thick,->,dotted] (c) to[bend right=10] node[right] {$x_{2}: 1$} (a);
\draw[thick,->,dotted] (a) to[bend right=10] node[left] {$y_1, y_{2}: 1$} (c);
\draw[thick,->,dotted] (d) to[bend right=10] node[right] {$x_1, x_{2}: 1$} (b);
\draw[thick,->,dotted] (b) to[bend right=10] node[left] {$y_{2}: 1$} (d);

\draw[fill=black] ($(a)+(-0.8,0.6)$) circle (0.05cm);
\draw[->,thick] ($(a)+(-0.8,0.6)$) -- (a);

\end{tikzpicture}
\caption{A minimal DPW for $\mathcal{L}(\mathcal{C}^2)$ with a marking of how the states map to combinations of states in a COCOA for the same language. Transitions with a color of $1$ are drawn as dotted.}
\label{fig:dpaForC2}
\end{figure}

We employ ideas from the study of \emph{rerailing automata} \cite{DBLP:journals/corr/abs-2503-08438}, which generalize deterministic parity automata. In particular, we study how strongly connected components in a parity automaton for $\mathcal{L}(\mathcal{C}^k)$ need to be \emph{nested}. The main observation used for proving Theorem~\ref{thm:conciseNessWithSingleSuffixLanguage} that can be obtained in this way is captured in the following lemma:

\begin{lem}
\label{lem:One}
Let $(Q',\delta')$ be a strongly connected component in $\mathcal{P}^k$ consisting only of reachable states and for some $1 \leq i < k$ and $1 \leq j < k$, we have that for any word $w$ in which only letters from $x_i, \ldots, x_{k+1}$ and $y_j, \allowbreak{} \ldots, \allowbreak{} y_{k+1}$ occur, a run for $w$ starting in any state in $Q'$ stays in $(Q',\delta')$.

Then, we have that there are disjoint reachable SCCs $(Q'^x,\delta'^x)$ and $(Q'^y,\delta'^y)$ within $(Q',\delta')$ s.t.
\begin{itemize}
\item for any word $w'$ with only letters from $x_{\max(i,j)+1} \ldots x_{k+1}$ and $y_j, \ldots, y_{k+1}$, any run from a state $q \in Q'^x$ for $w'$ stays in $(Q'^x,\delta'^x)$, and
\item for any word $w'$ with only letters from $x_{i} \ldots x_{k+1}$ and $ y_{\max(i,j)+1}, \ldots, y_{k+1}$, any run from a state $q \in Q'^y$ for $w'$ stays in $(Q'^y,\delta'^y)$.
\end{itemize}
\end{lem}
\begin{proof}
We can find the SCC $(Q'^x,\delta'^x)$ as follows:
Consider the set of transitions $T^x$ in $(Q',\delta')$ for letters from $x_{\max(i,j)+1} \ldots x_{k+1}, y_j, \ldots, y_{k+1}$. We use a subset of $T^x$ that forms a transition set of an SCC  as $\delta'^x$. Such a subset has to exist as all transitions from states in $Q'$ for letters in the considered letter set stay in $Q'$, and $Q'$ together with $T^x$ decomposes into SCCs.
We find the SCC $(Q'^y,\delta'^y)$ in the same way but for the letters $x_{i} \ldots x_{k+1}, y_{\max(i,j)+1}, \ldots, \allowbreak{} y_{k+1}$.

The SCCs $(Q'^x,\delta'^x)$ and $(Q'^y,\delta'^y)$ have the needed property: all outgoing transitions for letters in the considered character sets are within $\delta'^x$/$\delta'^y$, respectively, as they consist of all transitions for the respective characters within the SCCs, and due to how they were chosen, there are no outgoing transitions in $\mathcal{P}^k$ for the respective letter set.

To see that $(Q'^x,\delta'^x)$ and $(Q'^y,\delta'^y)$ are disjoint, 
consider first a word $w^x$ containing all letters from $x_{\max(i,j)+1} \ldots x_{k+1}, y_j, \ldots, y_{k+1}$ infinitely often and for which from some $q'^x \in Q'^x$, a run for $w^x$ takes all transitions in $\delta'^x$ infinitely often. Since  $(Q'^x,\delta'^x)$ is an SCC and contains transitions for all these letters, such a word has to exist.
Note that by the definition of $\mathcal{C}^k$, we have that $w^x$ is in the language of $\mathcal{C}^k$ if and only if $\max(i,j)$ is even.
We can build a similar word $w^y$ for $x_{i} \ldots x_{k+1}, y_{\max(i,j)+1}, \ldots, y_{k+1}$. It is also in the language of $\mathcal{C}^k$ if and only if $\max(i,j)$ is even.

If $(Q'^x,\delta'^x)$ and $(Q'^y,\delta'^y)$ would overlap, we could build a word/run combination $w^\mathit{mix}$/$\pi^\mathit{mix}$ from $w^x$ and $w^y$ by 
taking the prefix run/word of $w^x$ until reaching the joint state $q_\mathit{mix} \in Q'^x \cap Q'^y$, removing the stem of $w^y$ (i.e., the characters until when the respective run reaches $q_\mathit{mix}$), and then
switching between the words whenever $q_\mathit{mix}$ is reached along the run for $w^\mathit{mix}$. The resulting word $w^\mathit{mix}$ contains all letters from $x_{i} \ldots x_{k+1}, y_j, \ldots, y_{k+1}$ infinitely often and the run for the word takes all transitions in $\delta'^x \cup \delta'^y$ infinitely often. This means that the dominating color of the run of $w^\mathit{mix}$ is the least dominating color of runs induced by $w^x$ and $w^y$, respectively.

By the definition of $\mathcal{C}^k$, whether $w^\mathit{mix}$ is in the language of $\mathcal{C}^k$ needs to differ, however, from whether $w^x$ and $w^y$ are in the language of $\mathcal{C}^k$, as $w^\mathit{mix}$ is in $\mathcal{C}^k$ if and only if $\max(i,j)$ is odd. Hence, to avoid either $w^\mathit{mix}$, $w^x$, or $w^y$ to be recognized with a color that has the wrong evenness, we have that $Q'^x$ and $Q'^y$ need to be disjoint.
\end{proof}

This lemma can be used in an induction argument over the size of $\mathcal{P}^k$:
\begin{lem}
\label{lem:nofStates}
Let $(Q',\delta')$ be a strongly connected component in $\mathcal{P}$ such that from any state $q \in Q'$, for any word $w$ in which only letters from $x_i, \ldots, x_{k+1}, y_j, \ldots, y_{k+1}$ occur, a run for $w$ starting in $q$ stays in $(Q',\delta')$ (for some $1 \leq i \leq k$ and $1 \leq j \leq k$). %
We have that $Q'$ is of size at least $2^{k-\max(i,j)}$.
\end{lem}
\begin{proof}
We prove the claim by induction over $\max(i,j)$, starting from the case $\max(i,j)=k$ and progressing backwards.
For the induction basis ($\max(i,j)=k$), this claim is trivially true, as at least one state is needed in $(Q',\delta')$.

For the induction step, consider a concrete combination of $(i,j)$ with $i < k$ and $j<k$ (so that the induction basis does not apply). Lemma \ref{lem:One} states that there are distinct sub-SCCs $(Q'^x,\delta'^x)$ and $(Q'^y,\delta'^y)$ within $(Q',\delta')$ for letters from $x_{\max(i,j)+1} \ldots x_{k+1}, y_j, \ldots, y_{k+1}$ and $x_{i} \ldots x_{k+1}, y_{\max(i,j)+1}, \ldots, y_{k+1}$, respectively. By the induction hypothesis, these each have sizes of $2^{k-\max(i,j)-1}$. As $(Q',\delta')$ has both of these as distinct sub-SCCs, $Q'$ needs to have at least $2^{k-\max(i,j)}$ states.
\end{proof}
We are now ready to prove Theorem~\ref{thm:conciseNessWithSingleSuffixLanguage}. Note that it has not been proven yet that $\mathcal{C}^k$ is actually a canonical COCOA of the language it represents, which requires that every word is accepted with its natural color w.r.t.~the language of $\mathcal{C}^k$. Hence, the following proof starts with establishing this fact.
\begin{proof}[Proof of Theorem~\ref{thm:conciseNessWithSingleSuffixLanguage}]
We first prove that $\mathcal{C}^k$ is the COCOA of some language (for every $k \in \NN$).
To see this, consider first some word $w$ that is rejected by $A_1^k$. Then, 
both $x_1$ and $y_1$ appear in the word infinitely often. Injecting additional letters does not change that the word is rejected, and hence words rejected by $A_1^k$ have a natural color of $0$.

For the other automata, we show by induction that if an automaton ${A}^k_i$ is the one with smallest index accepting some word $w$, then the word has a natural color of $i$ w.r.t.~the language of $\mathcal{C}^k$. 
So let us assume that $w$ is accepted by $A_i^k$ but rejected by ${A}_{i+1}^k$ (if $i<k$). Then the word either contains $x_{i+1}$ infinitely often and all characters $x_1, \ldots, x_{i}$ only finitely often, or $y_{i+1}$ infinitely often  and all characters $y_1, \ldots, y_{i}$ only finitely often.
Any injection either maintain this property (hence keeping whether the word is in the language of $\mathcal{C}^k$) or injects characters from $x_1, \ldots, x_{i}$ or $y_1, \ldots, y_{i}$ infinitely often, and then the resulting word has a natural color that is strictly smaller.

For proving the size bound, applying Lemma~\ref{lem:nofStates} on $x_1, \ldots, x_{k+1}, y_1, \ldots, y_{k+1}$ and any SCC of $(Q,\delta)$ without outgoing edges yields that at least $2^{k}$ many states are needed for $\mathcal{P}^k$. Note that such an SCC always exists.
\end{proof}

We note that for the family of languages defined in this section, the size bound of Theorem~\ref{thm:conciseNessWithSingleSuffixLanguage} is actually tight, as by generalizing the construction depicted in Figure~\ref{fig:dpaForC2}, we can obtain parity automata $\mathcal{P}^k$ of size exactly $2^{k}$.

\section{COCOA disjunction/conjunction can cause an exponential blow-up}
\label{sec:conjunction}

We have seen in the previous section that COCOA can be exponentially more concise than deterministic parity automata. 
But how \emph{brittle} is this conciseness? In particular, can it be that a language can be represented concisely with COCOA (when compared to a DPW representation), but when processing the language, conciseness is shattered by the operation performed on the language? 
In turns out that this is indeed the case when considering conjunction and disjunction operations on COCOA, as we show in this section. 
We define two families of languages that can be concisely represented and prove that when taking their conjunction or disjunction, an exponential blow-up is unavoidable for this family. In contrast, conjunction or disjunction can be performed with polynomial blow-up when using a DPW representation for this family of languages.

{
We note that the blow-up is unrelated to any automaton size increase potentially caused by disjunction or conjunction operations on \hdcw{}, of which the COCOA are composed. 
Rather, the change in conciseness is caused by a restructuring of how the language to represent is mapped to the COCOA levels. We also note that in the general case, taking the conjunction or disjunction of DPWs has an unavoidable exponential blow-up \cite{DBLP:conf/lpar/Boker18}. %
}

We start by introducing the first family of languages $\{L^k\}_{k \in \NN}$ that have 
COCOA of size polynomial in $k$, but for which the number of residual languages is exponential in $k$. 

\begin{defi}
\label{def:toyLanguage}
Let $k \in \NN$ be given. We set $\Sigma = \{X_1, \ldots, X_k, Y_1, \ldots, Y_k, a_0, \ldots, a_{4k-1}\}$ and define $L^k = \mathcal{L}(L_1^k, \allowbreak{} \ldots, \allowbreak{} L_n^k)$ for the following sequence of language $L^k_1, \ldots, L^k_k$, where $1 \leq i \leq k$:
\begin{align*} 
L^k_i & = ((\Sigma \setminus \{ X_i \}) + X_i (\Sigma \setminus \{ X_i \})^* X_i )^* (a_0 + \ldots + a_{4k-2i+1} )^\omega + \Sigma^* (a_0 + \ldots + a_{4k-2i} )^\omega
\end{align*}
\end{defi}
Each language $L_i^k$ only includes words that eventually only contain lower-case letters. Which such words are in the language depends on which lower-case letters are infinitely often contained, and their order does not matter. If the number of $X_i$ letters at the beginning of the word is even, then the set of characters that can appear infinitely often in the word is slightly larger by also including $a_{4k-2i+1}$. For all $L^k_i$, the set of letters that may occur infinitely often is strictly larger than for $L^k_{i+1}$. Whether the number of $X_i$ letters in a word is even or odd is only relevant for $L^k_i$, but not for $L^k_j$ for $i \neq j$.
We will next show that
\begin{itemize}
\item each language $L^k_i$ is a co-Büchi language, and there exists a deterministic co-Büchi automaton for $L^k_i$ with 2 states, and
\item $\Sigma^\omega$ has words with natural colors of $0\ldots k$ and $L_i^k$ accepts exactly the words with a natural color of $i$ or more (w.r.t.~$L^k$) -- hence, the co-Büchi automata for $L^k_1, \ldots, L^k_k$ together form a valid COCOA.
\end{itemize}

\begin{lem}
\label{lem:gfgcoBuchiSizes}
Let $L^k_i$ be a language as defined in Def.~\ref{def:toyLanguage}. There exists a deterministic co-Büchi automaton with transition-based acceptance for $L^k_i$ with two states.
\end{lem}
\begin{proof}
The automaton shown in Figure~\ref{def:toyLanguageAutomaton} accepts the desired language.
\end{proof}

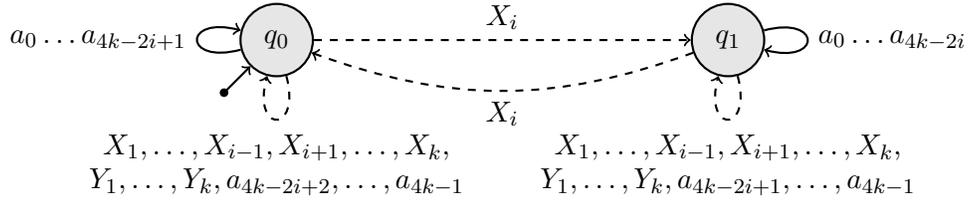
\begin{figure}
\centering
\begin{tikzpicture}

\node[state,thick,fill=black!10!white] (a) at (0,0) {$q_0$};

\node[state,thick,fill=black!10!white] (b) at (6,0) {$q_1$};

\draw[fill=black] ($(a)+(-0.7,-0.7)$) circle (0.05cm);
\draw[->,thick] ($(a)+(-0.7,-0.7)$) -- (a);

\draw[->,thick,dashed] (a) to[bend left=00] node[above] {$X_i$} (b);
\draw[->,thick,dashed] (b) to[bend left=20] node[below] {$X_i$} (a);
\draw[->,thick,dashed] (a) to[loop below] node[below] {
\begin{tabular}{c}$X_1, \ldots, X_{i-1}, X_{i+1}, \ldots, X_k$, \\ $Y_1, \ldots, Y_k, a_{4k-2i+2}, \ldots, a_{4k-1}$ \end{tabular}} (a);
\draw[->,thick] (a) to[loop left] node[left] {$a_0 \ldots a_{4k-2i+1}$} (a);

\draw[->,thick,dashed] (b) to[loop below] node[below] {
\begin{tabular}{c}$X_1, \ldots, X_{i-1}, X_{i+1}, \ldots, X_k$, \\ $Y_1, \ldots, Y_k, a_{4k-2i+1}, \ldots, a_{4k-1}$ \end{tabular}} (b);
\draw[->,thick] (b) to[loop right] node[right] {$a_0 \ldots a_{4k-2i}$} (b);
\end{tikzpicture}
\caption{A deterministic co-Büchi automaton for the language $L^k_i$. Rejecting transitions are dashed.}
\label{def:toyLanguageAutomaton}
\end{figure}

We note that there does not exist a smaller history-deterministic (or deterministic) automaton for the same language as there are two residual languages in the automaton: the language of state $q_0$ includes $(a_{4k-2i+1})^\omega$, whereas the language of state $q_1$ does not. As Abu Radi and Kupferman have shown that every co-Büchi language has a smallest \hdcw{} in which every state is labeled by its residual language \cite{DBLP:journals/lmcs/RadiK22}, we hence know that there does not exist a one-state \hdcw{} for this language.

An interesting property of the definition is that the overall $\omega$-regular language induced by the chain of languages has a number of residual languages that is exponential in the number of automaton states in the chain.

\begin{lem}
\label{lem:numberOfSuffixLanguages}
The number of residual languages of $L^k$ is $2^k$. In particular, there is one residual language for each combination of whether letter $X_i$ has been seen an even or odd number of times so far (for $1 \leq i \leq k$).
\label{lem:suffixLanguages}
\end{lem}
\begin{proof}
To prove the claim, we need to show that for each two prefix words $w'^1 \in \Sigma^*$ and $w'^2 \in \Sigma^*$, if $w'^1$ and $w'^2$ differ on the evenness of the number of $X_i$ characters (for some $1 \leq i \leq k$), there exists an infinite word $w \in \Sigma^\omega$ such that whether $w'^1 w \in L^k$ differs from whether $w'^2 w \in L^k$.

Let $1 \leq i \leq n$ be such that, without loss of generality, the number of $X_i$ characters in $w'^1$ is even while the number of $X_i$ characters in $w'^2$ is odd.
Then, $w = (a_{4k-2i+1})^\omega$ is such a word. By the definition of $L_i^k$ and $L_{i+1}^k$, we have that $w'^1 w \in L^k$ if and only if $i$ is even. By the definition of $L_i^k$ and $L_{i-1}^k$, we have that $w'^2 w \in L^k$ if and only if $i$ is odd. 
\end{proof}

Before defining the COCOA to combine the one for $L^k$ with, we need to prove that a sequence of co-Büchi automata for $L_1^k, \ldots, L_k^k$ is indeed a COCOA for $L^k$, i.e., that it recognizes each word with its natural color.

\begin{lem}
\label{lem:correctNaturalColors}
Let $k \in \NN$ and $L_1^k, \ldots, L_k^k$ be languages defined for $k \in \NN$ according to Def.~\ref{def:toyLanguage}. 
We have that every language $L^k_i$ contains exactly the words that have a natural color of $i$ regarding $L^k$.
\end{lem}
\begin{proof}

First of all, note that for every $i \in \NN$, we have that $L^k_{i} \supset L^k_{i+1}$ (if $i < k$). This is because $L^k_i$ contains all words ending with $(a_0 + \ldots + a_{4k-2i})^\omega$ whereas all words contained in $L^k_{i+1}$ need to end with $(a_0 + \ldots + a_{4k-2(i-1)+1})^\omega$. Thus, we have $L^k_1 \supset \ldots \supset L^k_k$.%

We now prove by induction that for each $1 \leq i \leq k$, the language $L^k_i$ contains exactly the words whose natural color regarding $L^k$ is at least $i$.

\emph{Induction basis:} Let $w$ be a word. We need to show that exactly the words $w \in L^k$ for which every sequence of injections leading to a word $w'$ that has $w' \in L^k$ as well are the ones \emph{not} in $L^k_1$. As all concerned languages only care about the letters occurring along a word, but not \emph{where} they occur in a word, we do not have to reason over the set of indices at which residual language invariant words are injected into $w$.

$\Rightarrow$: So let $w$ be a word not in $L^k_1$. We need to show that then for a word $w'$ that extends $w$, we have $w' \in L^k$ as well.
We split on the possible reasons for why we can have $w \notin L^k_1$.
\begin{itemize}
\item If there are infinitely many letters from $X_1, \ldots, X_k, Y_1, \ldots, Y_k$ in $w$, then injecting additional letters does not change this, and hence $w'$ is not in $L^k_1$ as well, which implies that $w' \in L^k$.
\item If there are infinitely many letters $a_{4k-1}$ in $w$ and the previous case does not hold, then there is an odd number of $A_1$ letters in $w$. In $w'$,  either infinitely many letters from $X_1, \ldots, X_k, Y_1, \ldots, Y_k$ are injected, and then $w'$ is in $L^k$ by the reasoning above, or finitely many such letters are injected, but an even number of $A_1$ letters (as otherwise some injection is not residual language invariant). In this case, because then $a_{4k-1}$ is still infinitely often in $w'$, and an odd number of $X_i$ letters is in $w'$, we also have that $w' \notin L^k$.
\end{itemize}
As all other words are in $L^k_1$, this direction of the induction basis is finished.

$\Leftarrow$: Let $w$ be a word in $L^k$ such that for every every word $w'$ that extends $w$, we have $w' \in L^k$ as well. We have to show that then, $w \notin L^k_1$ holds.

Let $w$ be such a word. If $w$ contains infinitely many $X_i$ or $Y_i$ letters (for some $1 \leq i \leq k$), then $w$ and $w'$ are not in $L^k_1$. In all other cases, we can assume that there is an odd number of $X_1$ letters in the word and there are infinitely many $a_{4k-1}$ letters in the word as otherwise by injecting infinitely many $a_{4k-2}$ letters into $w$, we can ensure that $w' \in L^k_1$ but not $w' \in L^k_2$, which contradicts the premise. Such words $w$ are also rejected by $L^k_1$, which completes this part of the induction basis.

\emph{Induction step:}
In the induction step, we have to prove for $i \geq 2$ that under the assumption that $L_{i-1}$ accepts exactly the words with a natural color of at least $i-1$ (w.r.t~$L^k$), $L^k_i$ does not contain the words $w$ with a natural color of less than $i$. This is equivalent to stating that $L^k_i$ rejects exactly the words $w$ for which for every extension, either (a) the resulting word $w'$ is in $L^k$ if and only if $w$ is, or (b) the resulting word has a natural color of less than $i$. 

$\Rightarrow$: Let $w \notin L^k_i$. We need to show that for every extension of $w$, either (a) the resulting word $w'$ is in $L^k$ if and only if $w$ is, or (b) the resulting word has a natural color of less than $i$. 

Let such a word $w$ be given, and consider a word $w'$. We perform a case split:
\begin{itemize}
\item If $w'$ contains infinitely many letters of the form $X_i$ or $Y_i$ (for some $1 \leq i \leq k$), then it is not accepted by $L^k_i$ nor by $L^k_1$. By the inductive hypothesis, the claim holds in this case as then the word has a natural color of $0$.
\item If $w'$ does not contain infinitely many letters of the form $X_i$ or $Y_i$ (for some $1 \leq i \leq k$), then so does $w$. Let $a_h$ be the letter with the highest index $h$ occurring infinitely often in $w$ and $a_{h'}$ be the letter with highest index $h'$ occurring infinitely often in $w'$. Every capital letter $X_i$ can only be injected an even number of times as otherwise $w'$ would not extend $w$ (see Lemma~\ref{lem:suffixLanguages}). In this case, as injecting additional $a_l$ letters for some $0 \leq l \leq 4k-1$ cannot make a word not in $L^k_i$ contained in $L^k_i$, we have by the inductive hypothesis that $w'$ either has a natural color that is lower than $i$, or we have that $w' \in L^k_{i-1}$, and then we have that $w \in L^k$ if and only if $w' \in L^k$.
\end{itemize}

$\Leftarrow$: Let $w$ be a word such that for every word $w'$ resulting from residual language invariant word injections, we have that either $w'$ is in $L^k$ if and only if $w \in L^k$, or (b) the resulting word has a natural color of strictly less than $i$. We need to show that then, $w \notin L^k_i$.

Let $w$ be such a word. If $w$ contains infinitely many $X_i$ or $Y_i$ letters for some $1 \leq i \leq k$, then it is in $L^k$, and so is $w'$. For all other words, we do a case split on the highest index $h$ such that $w$ contains infinitely many letters $a_h$.

\begin{itemize}
\item If $h \geq 4k-2i+2$, then $L^k_{i}$ neither contains $w$ nor $w'$.
\item If $h = 4k-2i+1$ and the number of $X_i$ letters in $w$ is odd, then $L^k_{i}$ does not contain $w$ or $w'$ either.
\item Assume that $h = 4k-2i+1$ and the number of $X_i$ letters in $w$ is even, then we can extend $w$ to $w'$ by infinitely many $a_{h+1}$ letters. 
This makes $w \in L^k_{i}$ and $w' \notin L^k_{i}$. Since $\{w,w'\} \subseteq L^k_{i-1}$, we have that $(w \in L^k) \neq (w' \in L^k)$. However, as $w' \in L^k_{i-1}$, by the inductive hypothesis, we cannot have that the natural color of $w'$ is strictly less than $i$.
Hence, the assumption that $h = 4k-2+1$ and the number number of $X_i$ letters in $w$ is even has to be incorrect.
\item Assume that $h \leq 4k-2i$. Then we can inject infinitely many letters $a_{h+1}$ into $w$ if $h$ is odd, and $a_{h+2}$ if $h$ is even. By the definition of the languages $L^k_1, \ldots, L^k_k$ and $L^k$, these injections change whether the word is in $L^k$. However, both $w$ and its extension are in $L^k_{i-1}$ in this case, which contradicts the premise that every extension of $w$ that differs from $w$ w.r.t.~its containment in $L^k$ has a natural color that is less than $i$. \qedhere
\end{itemize}
\end{proof}
Let us now define the family of languages to combine the COCOA for $L^k$ with.

\begin{defi}
\label{def:toyLanguageB}
Let $k \in \NN$ be given. We set $\Sigma = \{X_1, \ldots, X_k, Y_1, \ldots, Y_k, a_0, \ldots, a_{4k-1}\}$ 
and define $\hat L^k = \mathcal{L}(\hat L^k_1, \allowbreak{} \ldots, \allowbreak{} \hat L^k_n)$ for the following sequence of languages, where $1 \leq i \leq k$:
\begin{align*} 
\hat L^k_i & = ((\Sigma \setminus \{ Y_i \}) + Y_i (\Sigma \setminus \{ Y_i \})^* Y_i )^* (a_{2i-2} + \ldots + a_{4k-1} )^\omega + \Sigma^* (a_{2i-1} + \ldots + a_{4k-1} )^\omega
\end{align*}
\end{defi}

Note that the properties of $L^k$ established in Lemma~\ref{lem:gfgcoBuchiSizes} and Lemma~\ref{lem:correctNaturalColors} carry over to $\hat L^k$ as well, as the languages only differ by swapping the roles of the letters $\{X_i\}_{1 \leq 1 \leq k}$ and $\{Y_i\}_{1 \leq 1 \leq k}$ as well as {swapping the letters $a_i$ and  $a_{4k-i-1}$ for each $0 \leq i < 2k$}.

Let in the following $L'^k = L^k \cap \hat L^k$. 
We will next analyze how big a COCOA for $L'^k$ needs to be and in this way shed light on how big the conjunction of COCOA for $L^k$ and $\hat L^k$ need to be.
To perform this analysis, we consider the language intersections $L^k_i \cap \hat L^k_j$ (for $1 \leq i \leq k$ and $1 \leq j \leq k$) and show how a COCOA for $L'^k$ can be built from disjunctions of some co-Büchi automata for $L^k_i \cap \hat L^k_j$.

\begin{lem}
\label{lem:movingLemma}
Let $w \in \Sigma^\omega$ be a word. There exists a unique greatest index pair $(i,j) \in \{0, \ldots, k\}^2$ such that
$w \in L^k_i \cap \hat L^k_j$, i.e., we have $w \in L^k_i \cap \hat L^k_j$ and for all $(i',j') \in \{0, \ldots, k\}^2$ such that $w \in L^k_{i'} \cap \hat L^k_{j'}$, we have that $i' \leq i$ and $j' \leq j$.

Furthermore, for every pair $(i', j')$ with $i' \leq i$ and $j' \leq j$, there exists an extension $w'$ of $w$ such that $(i', j')$ is the unique greatest index pair such that $w' \in L^k_{i'} \cap \hat L^k_{j'}$.
\end{lem}
\begin{proof}
For the first half, first of all note that $w \in L^k_0 \cap \hat L^k_0$ by definition as both $L^k_0$ and $\hat L^k_0$ {contain all infinite words over $\Sigma$.}
Then, let $K$ be the set of elements $(i,j)$ such that we have $w \in L^k_i \cap \hat L^k_j$.
If we have $(i,j) \in K$ and $(i',j') \in K$ for some such pairs, this means that $w \in L^k_i$, $w \in \hat L^k_j$, $w \in L^k_{i'}$, and $w \in \hat L^k_{j'}$, so we then also have $(\max(i,i'),\max(j,j')) \in K$. So we cannot have that both $(i,j)$ and $(i',j')$ are incomparable (with respect to element-wise comparison) maximal elements in $K$, as otherwise $(\max(i,i'),\max(j,j'))$ is another element in $K$, contradicting the assumption that both $(i,j)$ and $(i',j')$ are incomparable maximal elements of $K$. 

For the second half of the claim, let $w$ be given, let $(i,j)$ be the (unique) maximal level in $K$, and $(i',j')$ be such that $i' \leq i$ and $j' \leq j$. By injecting infinitely often $a_{4k-2i'}$ into $w$, the resulting word is in {$L^k_{i'}$} (but not in $L^k_{i'+1}$), and by injecting infinitely often $a_{2j'-1}$, the resulting word is in $\hat L^k_{j'}$ (but not in $\hat L^k_{j'+1}$). Definitions~\ref{def:toyLanguage} and \ref{def:toyLanguageB} are such that the former letter injections do not affect where in the chain $\hat L^k_1, \ldots, \hat L^k_k$ the resulting word is located, while the latter letter injections do not affect where in the chain $L^k_1, \ldots, L^k_k$ the resulting word is located. Hence, {the extended word has $(i',j')$ as the unique greatest index pair}.
\end{proof}

\begin{thm}
\label{thm:conjunctionBecomesBig}
A COCOA for $L'^k = L^k \cap \hat L^k$ can be given as $\mathcal{C}^{k} = (C^k_1, \ldots, C^k_{2k})$ where for each $u \in \{0, \allowbreak{} \ldots, \allowbreak{}2k\}$, the language of $C^k_u$ is 
\begin{equation*}
\mathcal{L}(C^k_{u}) = \bigcup_{(i,j) \in \Gamma_u} L^k_i \cap \hat L^k_j
\end{equation*}
for
\begin{equation*}
\Gamma^k_u = \begin{cases}
\{(i,j) \in \{0, \ldots, k\}^2 \mid i+j=u, i \text{ is even},j \text{ is even} \} & \text{if } u \in \{0, 2, \ldots, 2k\} \\
\{(i,j) \in \{0, \ldots, k\}^2 \mid u \leq i+j \leq u+1, i \text{ or } j \text{ are odd} \} & \text{if } u \in \{1, 3, \ldots, 2k-1\}.
\end{cases}
\end{equation*}
\end{thm}
Figure~\ref{fig:levels} shows how the languages $\{ L^k_i \cap \hat L^k_j\}_{0 \leq i \leq k, 0 \leq j \leq k}$ are grouped (by performing language disjunctions) to form a COCOA for $L'^k$. Languages $L^k_i \cap \hat L^k_j$ in which both $i$ and $j$ are even form the accepting levels of the COCOA, and the languages in between are grouped into rejecting levels of the COCOA for $L'^k$.

\begin{figure}

\resizebox{0.99\columnwidth}{!}{
\begin{tikzpicture}[yscale=0.9]

\path[fill=red!10!white] (-1.5,0) ..controls +(0.75,0) and +(-0.75,0) .. (0,-0.5) ..controls +(0.75,0) and +(-0.75,0) .. (1.5,0) ..controls +(-1,1) and +(1,1) .. cycle;

\path[fill=red!10!white]   (-4.5,-2) ..controls +(0.75,0) and +(-0.75,0) .. (-3,-2.5) ..controls +(0.75,0) and +(-0.75,0) .. (-1.5,-2) ..controls +(0.75,0) and +(-0.75,0) .. (0,-2.75) ..controls +(0.75,0) and +(-0.75,0) .. (1.5,-2) ..controls +(0.75,0) and +(-0.75,0) .. (3,-2.5) ..controls +(0.75,0) and +(-0.75,0) .. (4.5,-2)
-- 
(3,-1) 
..controls +(-0.75,0) and +(0.75,0) ..
(1.5,-1.5)
..controls +(-0.75,0) and +(0.75,0) ..
(0,-2.5)
..controls +(-0.75,0) and +(0.75,0) ..
(-1.5,-1.5)
..controls +(-0.75,0) and +(0.75,0) ..
(-3.0,-1)
;

\path[fill=red!10!white] (-6,-3) 
..controls +(0.75,0) and +(-0.75,0) .. 
(-4.5,-3.5)
..controls +(0.75,0) and +(-0.75,0) ..
(-3,-4.5)
..controls +(0.75,0) and +(-0.75,0) ..
(-1.5,-3.5)
--
(1.5,-3.5)
..controls +(0.75,0) and +(-0.75,0) ..
(3.0,-4.5)
..controls +(0.75,0) and +(-0.75,0) ..
(4.5,-3.5)
..controls +(0.75,0) and +(-0.75,0) ..
(6.0,-3)
--
(7.5,-4)
..controls +(-0.75,0) and +(0.75,0) ..
(6,-4.5)
..controls +(-0.75,0) and +(0.75,0) ..
(4.5,-4)
..controls +(-0.75,0) and +(0.75,0) ..
(3.0,-4.75)
..controls +(-0.75,0) and +(0.75,0) ..
(1.5,-4.5)
--
(-1.5,-4.5)
..controls +(-0.75,0) and +(0.75,0) ..
(-3,-4.75)
..controls +(-0.75,0) and +(0.75,0) ..
(-4.5,-4) 
..controls +(-0.75,0) and +(0.75,0) ..
(-6,-4.5) 
..controls +(-0.75,0) and +(0.75,0) ..
(-7.5,-4)
-- cycle
;

\path[fill=red!10!white] (-6,-5)
..controls +(0.75,0) and +(-0.75,0) ..
(-4.5,-5.5)
--
(-1.5,-5.5)
..controls +(0.75,0) and +(-0.75,0) ..
(0,-6.5)
..controls +(0.75,0) and +(-0.75,0) ..
(1.5,-5.5)
--
(4.5,-5.5)
..controls +(0.75,0) and +(-0.75,0) ..
(6,-5.25)
--
(4.5,-6.5)
..controls +(-0.75,0) and +(0.75,0) ..
(3,-6.5)
..controls +(-0.75,0) and +(0.75,0) ..
(1.5,-6.0)
..controls +(-0.75,0) and +(0.75,0) ..
(0,-6.75)
..controls +(-0.75,0) and +(0.75,0) ..
(-1.5,-6.0)
..controls +(-0.75,0) and +(0.75,0) ..
(-3,-6.5)
..controls +(-0.75,0) and +(0.75,0) ..
(-4.5,-6)
-- cycle
;

\path[fill=red!10!white] (-3,-7)
..controls +(0.75,0) and +(-0.75,0) ..
(-1.5,-7.5)
--
(1.5,-7.5)
..controls +(0.75,0) and +(-0.75,0) ..
(3,-7.25) 
..controls +(-2,-2) and +(2,-2) ..
(-3,-7)
;

\node[inner sep=2pt] (a00) at (0,0) {$L^4_0 \cap \hat L^4_0$};

\node[inner sep=2pt] (a01) at (-1.5,-1) {$L^4_0 \cap \hat L^4_1$};
\node[inner sep=2pt] (a10) at (1.5,-1) {$L^4_1 \cap \hat L^4_0$};

\node[inner sep=2pt] (a02) at (-3,-2) {$L^4_0 \cap \hat L^4_2$};
\node[inner sep=2pt] (a11) at (0,-2) {$L^4_1 \cap \hat L^4_1$};
\node[inner sep=2pt] (a20) at (3,-2) {$L^4_2 \cap \hat L^4_0$};

\node[inner sep=2pt] (a03) at (-4.5,-3) {$L^4_0 \cap \hat L^4_3$};
\node[inner sep=2pt] (a12) at (-1.5,-3) {$L^4_1 \cap \hat L^4_2$};
\node[inner sep=2pt] (a21) at (1.5,-3) {$L^4_2 \cap \hat L^4_1$};
\node[inner sep=2pt] (a30) at (4.5,-3) {$L^4_3 \cap \hat L^4_0$};

\node[inner sep=2pt] (a04) at (-6,-4) {$L^4_0 \cap \hat L^4_4$};
\node[inner sep=2pt] (a13) at (-3,-4) {$L^4_1 \cap \hat L^4_3$};
\node[inner sep=2pt] (a22) at (0,-4) {$L^4_2 \cap \hat L^4_2$};
\node[inner sep=2pt] (a31) at (3,-4) {$L^4_3 \cap \hat L^4_1$};
\node[inner sep=2pt] (a40) at (6,-4) {$L^4_4 \cap \hat L^4_0$};

\node[inner sep=2pt] (a14) at (-4.5,-5) {$L^4_1 \cap \hat L^4_4$};
\node[inner sep=2pt] (a23) at (-1.5,-5) {$L^4_2 \cap \hat L^4_3$};
\node[inner sep=2pt] (a32) at (1.5,-5) {$L^4_3 \cap \hat L^4_2$};
\node[inner sep=2pt] (a41) at (4.5,-5) {$L^4_4 \cap \hat L^4_1$};

\node[inner sep=2pt] (a24) at (-3,-6) {$L^4_2 \cap \hat L^4_4$};
\node[inner sep=2pt] (a33) at (0,-6) {$L^4_3 \cap \hat L^4_3$};
\node[inner sep=2pt] (a42) at (3,-6) {$L^4_4 \cap \hat L^4_2$};

\node[inner sep=2pt] (a34) at (-1.5,-7) {$L^4_3 \cap \hat L^4_4$};
\node[inner sep=2pt] (a43) at (1.5,-7) {$L^4_4 \cap \hat L^4_3$};

\node[inner sep=2pt] (a44) at (0,-8) {$L^4_4 \cap \hat L^4_4$};

\draw[thick] (a00) -- (a01) -- (a02) -- (a03) -- (a04) -- (a14) -- (a24) -- (a34) -- (a44);
\draw[thick] (a00) -- (a10) -- (a20) -- (a30) -- (a40) -- (a41) -- (a42) -- (a43) -- (a44);

\draw[thick] (a10) -- (a11) -- (a12) -- (a13) -- (a14);
\draw[thick] (a20) -- (a21) -- (a22) -- (a23) -- (a24);
\draw[thick] (a30) -- (a31) -- (a32) -- (a33) -- (a34);
\draw[thick] (a01) -- (a11) -- (a21) -- (a31) -- (a41);
\draw[thick] (a02) -- (a12) -- (a22) -- (a32) -- (a42);
\draw[thick] (a03) -- (a13) -- (a23) -- (a33) -- (a43);

\draw[thick,color=red] (-1.5,0) ..controls +(0.75,0) and +(-0.75,0) .. (0,-0.5) ..controls +(0.75,0) and +(-0.75,0) .. (1.5,0) node[above left,yshift=4pt,xshift=12pt,fill=red!10!white,shape=rectangle,rounded corners=0.2cm,inner sep=3pt] {$\mathcal{L}(C^4_0)$};

\draw[thick,color=red] (-3.0,-1) ..controls +(0.75,0) and +(-0.75,0) .. (-1.5,-1.5) ..controls +(0.75,0) and +(-0.75,0) .. (0,-2.5) ..controls +(0.75,0) and +(-0.75,0) .. (1.5,-1.5) ..controls +(0.75,0) and +(-0.75,0) .. (3,-1) 
node[above left,shape=rectangle,rounded corners=0.2cm,inner sep=3pt,xshift=12pt] {$\mathcal{L}(C^4_1)$};

\draw[thick,color=red] (-4.5,-2) ..controls +(0.75,0) and +(-0.75,0) .. (-3,-2.5) ..controls +(0.75,0) and +(-0.75,0) .. (-1.5,-2) ..controls +(0.75,0) and +(-0.75,0) .. (0,-2.75) ..controls +(0.75,0) and +(-0.75,0) .. (1.5,-2) ..controls +(0.75,0) and +(-0.75,0) .. (3,-2.5) ..controls +(0.75,0) and +(-0.75,0) .. (4.5,-2) node[above left,yshift=4pt,fill=red!10!white,shape=rectangle,rounded corners=0.2cm,inner sep=3pt,xshift=12pt] {$\mathcal{L}(C^4_2)$};

\draw[thick,color=red] (-6,-3) 
..controls +(0.75,0) and +(-0.75,0) .. 
(-4.5,-3.5)
..controls +(0.75,0) and +(-0.75,0) ..
(-3,-4.5)
..controls +(0.75,0) and +(-0.75,0) ..
(-1.5,-3.5)
--
(1.5,-3.5)
..controls +(0.75,0) and +(-0.75,0) ..
(3.0,-4.5)
..controls +(0.75,0) and +(-0.75,0) ..
(4.5,-3.5)
..controls +(0.75,0) and +(-0.75,0) ..
(6.0,-3)
node[above left,shape=rectangle,rounded corners=0.2cm,inner sep=3pt,xshift=12pt] {$\mathcal{L}(C^4_3)$}
;

\draw[thick,color=red] (-7.5,-4)
..controls +(0.75,0) and +(-0.75,0) ..
(-6,-4.5) 
..controls +(0.75,0) and +(-0.75,0) ..
(-4.5,-4) 
..controls +(0.75,0) and +(-0.75,0) ..
(-3,-4.75)
..controls +(0.75,0) and +(-0.75,0) ..
(-1.5,-4.5)
--
(1.5,-4.5)
..controls +(0.75,0) and +(-0.75,0) ..
(3.0,-4.75)
..controls +(0.75,0) and +(-0.75,0) ..
(4.5,-4)
..controls +(0.75,0) and +(-0.75,0) ..
(6,-4.5)
..controls +(0.75,0) and +(-0.75,0) ..
(7.5,-4)
node[above left,yshift=4pt,fill=red!10!white,shape=rectangle,rounded corners=0.2cm,inner sep=3pt,xshift=12pt] {$\mathcal{L}(C^4_4)$}
;

\draw[thick,color=red] (-6,-5)
..controls +(0.75,0) and +(-0.75,0) ..
(-4.5,-5.5)
--
(-1.5,-5.5)
..controls +(0.75,0) and +(-0.75,0) ..
(0,-6.5)
..controls +(0.75,0) and +(-0.75,0) ..
(1.5,-5.5)
--
(4.5,-5.5)
..controls +(0.75,0) and +(-0.75,0) ..
(6,-5.25)
node[above left,shape=rectangle,rounded corners=0.2cm,inner sep=3pt,xshift=20pt] {$\mathcal{L}(C^4_5)$}
;

\draw[thick,color=red] (-4.5,-6)
..controls +(0.75,0) and +(-0.75,0) ..
(-3,-6.5)
..controls +(0.75,0) and +(-0.75,0) ..
(-1.5,-6.0)
..controls +(0.75,0) and +(-0.75,0) ..
(0,-6.75)
..controls +(0.75,0) and +(-0.75,0) ..
(1.5,-6.0)
..controls +(0.75,0) and +(-0.75,0) ..
(3,-6.5)
..controls +(0.75,0) and +(-0.75,0) ..
(4.5,-6.5)
node[above left,yshift=4pt,fill=red!10!white,shape=rectangle,rounded corners=0.2cm,inner sep=3pt,xshift=20pt] {$\mathcal{L}(C^4_6)$}
;

\draw[thick,color=red] (-3,-7)
..controls +(0.75,0) and +(-0.75,0) ..
(-1.5,-7.5)
--
(1.5,-7.5)
..controls +(0.75,0) and +(-0.75,0) ..
(3,-7.25)
node[above left,shape=rectangle,rounded corners=0.2cm,inner sep=3pt,xshift=20pt] {$\mathcal{L}(C^4_7)$}

node[below left,xshift=-5pt,yshift=-5pt,fill=red!10!white,shape=rectangle,rounded corners=0.2cm,inner sep=3pt,xshift=12pt] {$\mathcal{L}(C^4_8)$}
;

\end{tikzpicture}}
\caption{Overview of how the sets $\{ L^k_i \cap \hat L^k_j\}_{0 \leq i \leq k, 0 \leq j \leq k}$ (for $k=4$) compose $C^k_0, \ldots, C_{2k}^k$ in Theorem~\ref{thm:conjunctionBecomesBig}} 
\label{fig:levels}
\end{figure}
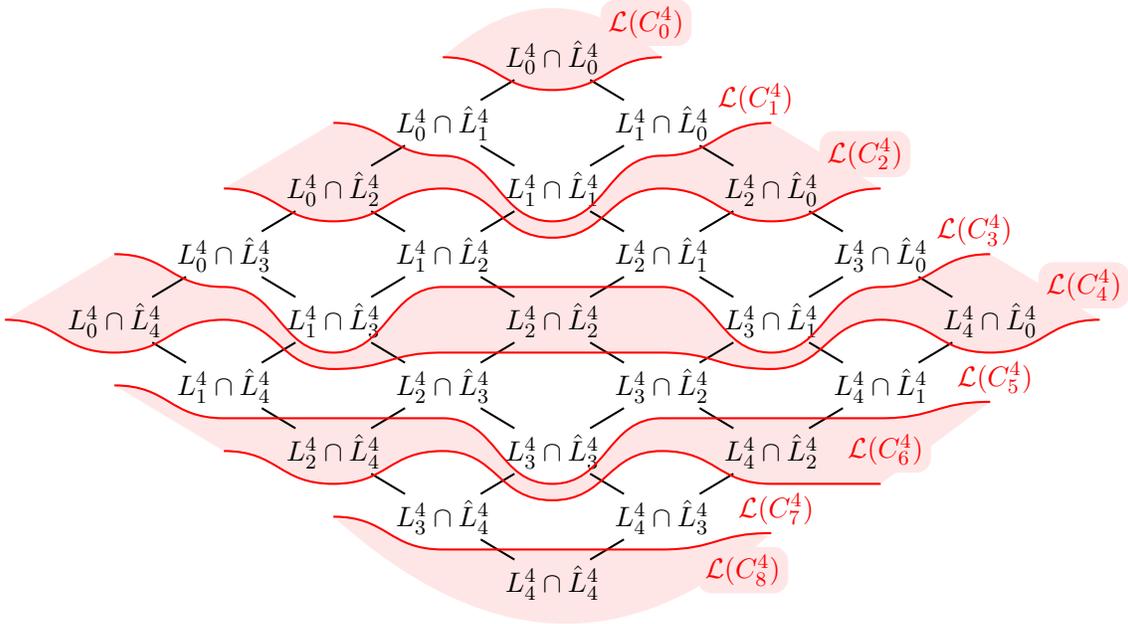

The correctness proof of Theorem~\ref{thm:conjunctionBecomesBig} employs a couple of observations:

\begin{obs}
\label{obs:three}
Let $w$ be some word for which the unique greatest pair $(i,j)$ such that $w \in L^k_i \cap \hat L^k_j$ holds (with $(i,j) \in \Gamma^k_u$ for some $0 \leq u \leq 2k$). We have that $w \notin \mathcal{L}(C^k_{u+1})$.
\end{obs}
\begin{proof}
If $(i,j)$ is the unique greatest pair $(i,j)$ such that $w \in L^k_i \cap \hat L^k_j$, then we have that $w \notin L^k_{i+1}$ and $w \notin \hat L^k_{j+1}$.

We consider two cases, namely that $u$ is even and that $u$ is odd.

In the first case, we have that $i+j = u$, both $i$ and $j$ are even, and we are searching for some other $(i',j')$ in $\Gamma^k_{u+1}$ such that $i' \leq i$ (so that $w$ is accepted by $\hat L^k_{i'}$) and such that $j' \leq j$ (so that $w$ is accepted by $L^k_{j'}$). If $u$ is even, then $u+1$ is odd, and by the definition of $\{\Gamma^k_m\}_{0 \leq m \leq 2k}$ in Theorem~\ref{thm:conjunctionBecomesBig}, we have $u+1 \leq i'+j' \leq u+2$ for $(i',j')$ to be contained in $\Gamma^k_{u+1}$. Since $i' \leq i$, $j' \leq j$, and $i'+j' \geq i+j+1$, we have that $w$ is not accepted by $C^k_{u+1}$ as not all three requirements on $i'$ and $j'$ can be fulfilled at the same time.

For the case that $u$ is odd, we have $u \leq i+j \leq u+1$ and either $i$ or $j$ are odd. Here, we are searching for $i' \leq i$, $j' \leq j$ with $u+1 = i'+j'$ and $i'$ and $j'$ are even. %
This leaves only $i=i'$ and $j=j'$ as solution, which however contradicts thats either $i$ or $j$ are odd, as $i'$ and $j'$ are even. Hence, we have that $w$ is not accepted by $C^k_{u+1}$.
\end{proof}

\begin{obs}
\label{obs:four}
Whenever for some $0 < i \leq k$, $0 \leq j \leq k$, and $0 \leq u \leq 2k$ for even $i$ and $j$, we have that $(i-1,j) \in \Gamma^k_{u-1}$, then $(i,j) \in \Gamma^k_u$. 

Similarly, whenever for some $0 \leq i \leq k$, $0 \leq j \leq k$, and $0 \leq u \leq 2k$ for even $i$ and $j$, we have that $(i,j-1) \in \Gamma^k_{u-1}$, then $(i,j) \in \Gamma^k_u$. 
\end{obs}
\begin{proof}
Let $(i-1,j) \in \Gamma^k_{u-1}$ for even $i$ and $j$. Then, $i-1$ is odd,  $j$ is even, and $u-1$ is odd. So $u-1 \leq i-1+j \leq u$. This means that $u=i+j$, and then, by the definition of $\{\Gamma^k_m\}_{0 \leq m \leq 2k}$ in Theorem~\ref{thm:conjunctionBecomesBig}, we have that $(i,j)$ is in $\Gamma^k_u$.

The other case can be proven in an analogous way.
\end{proof}

\begin{obs}
\label{obs:five}
Let $(i,j) \in \Gamma^k_u$ and $(i',j')$ be some pair such that $0 \leq i' \leq i$, $0 \leq j' \leq j$, $(i,j) \neq (i',j')$, and both $i'$ and $j'$ are even. Then we have that $(i',j') \in \Gamma^k_0 \cup \ldots \cup \Gamma^k_{u-1}$. 
\end{obs}
\begin{proof}
Let $u'$ be the element such that $(i',j') \in \Gamma^k_{u'}$.

If $u$ is even, both $i$ and $j$ are both even, and $u = i+j < i'+j' = u'$, so that $(i',j') \in \Gamma^k_0 \cup \ldots \cup \Gamma^k_{u-1}$.

If $u$ is odd, then $u \leq i +j \leq u+1$. We distinguish two cases, namely that $i+j = u$ and that $i+j = u+1$.

In the first case, since $u'=i'+j'$ (as $u'$ is even), $i' \leq i$, $j' \leq j$, and $(i,j) \neq (i',j')$, we have that $u' < u$, which is sufficient.

In the second case, we have $i+j = u+1$. Since either $i$ or $j$ are odd and $u$ is odd, we have that both $i$ and $j$ have to be odd for $u+1$ to be even. Since $i' \leq i$, $j' \leq j$, and both $i'$ and $j'$ are even, we have that $i' < i$ and $j' < j$. So we have $u' = i'+j' \leq i+j -2 = u-1$. Thus, we have $(i',j') \in \Gamma^k_0 \cup \ldots \cup \Gamma^k_{u-1}$ again.
\end{proof}

\begin{proof}[Proof of Theorem \ref{thm:conjunctionBecomesBig}]
We prove the claim by induction over $u$. 
In particular, we show that exactly the words with a natural color of $u$ are rejected by $C^k_{u+1}$ but accepted by $C^k_{u}$.
Note that this is equivalent to showing that for each word $w$ with a natural color of $u$, for the unique maximal pair $(i,j)$ with $0 \leq i \leq k$ and $0 \leq j \leq k$ such that $w \in L^k_i \cap \hat L^k_j$, we have that $(i,j) \in \Gamma_u^k$.

We split proving the claim into four cases, and these are the possible combinations of even $u$/odd $u$ and the two directions of the proof.

\textbf{Even $u$, $\Leftarrow$:} 
Consider first an even $u \geq 0$ and let $w$ be a word accepted by $C^k_u$ but rejected by $C^k_{u+1}$. We have that the unique greatest pair $(i,j)$ with $w \in L^k_i \cap \hat L^k_j$ is such that $i+j = u$ with even $i$ and $j$ (as $C^k_u$ is built as a disjunction of automata for $L^k_i \cap \hat L^k_{j}$ from such pairs). 

Since $i$ and $j$ are both even, we have that $w$ is both in  $L^k$ as well as $\hat L^k$, and hence should be in the language of $\mathcal{C}^k$. Since $u$ is even and $w$ is accepted by $C^k_u$ and rejected by $C^k_{u+1}$, this is indeed the case.

Let us now show that the conditions for $w$ to have a natural color of $u$ are fulfilled. In particular, 
let now $w'$ be an extension of $w$ and $(i',j')$ be the  unique greatest pair such that $w' \in L^k_{i'} \cap \hat L^k_{j'}$. By Lemma~\ref{lem:movingLemma}, we have that $i' \geq i$ and $j' \geq j$.
If $i=i'$ and $j=j'$, then the resulting word is still in $L^k$, $\hat L^k$, and the language of $\mathcal{C}^k$, which satisfies the conditions for the natural color of $w$ to be $u$. If however $i'<i$ or $j'<j$, then the resulting word is rejected by $C^k_u$, and by the inductive hypothesis has a smaller natural color, which is also fine for the word $w$ to have a natural color of $u$.

\textbf{Even $u$, $\Rightarrow$:} 
Let $w$ be a word with a natural color of $u$ and $u$ be even. Since $u$ is even, the word is in $L'^k$ and hence there exists some unique greatest pair $(i,j)$ with even $i$ and $j$ such that $w \in L^k_i \cap L^k_j$. We show that $(i,j) \in \Gamma^k_u$. 

We first distinguish between $u$ being $0$ or not. If $u=0$, then there does not exist some level $(i',j')$ with odd $i'$ or odd $j'$ and such that $i' \leq i$ and $j' \leq j$, as by Lemma~\ref{lem:movingLemma}, we could then compute a word $w'$ not in $L'^k$ that is an extension of $w$, which contradicts that $w$ has a natural color of $0$. This means that $i=0$ and $j=0$, as otherwise an extension to a word $w'$ for which $(0,1)$ or $(1,0)$ are the unique greatest pairs for $w'$ would be possible (by Lemma~\ref{lem:movingLemma}). Hence, $w$ has $(0,0)$ as unique greatest pair, and this pair is contained in $\Gamma^k_0$.

For the case of $u\geq 2$, 
we either have $i > 0$ or $j > 0$.
Let us consider the case that $i>0$. Then there exists an extension $w'$ of $w$ whose greatest pair is $(i-1,j)$. Since the natural color of $w$ is $u$, by assumption and because $w' \notin L^k$, the natural color of $w'$ is at most $u-1$. Then by the inductive hypothesis, we have $(i-1,j) \in \Gamma^k_0 \cup \ldots \cup \Gamma^k_{u-1}$. If $(i-1,j) \in \Gamma^k_0 \cup \ldots \Gamma^k_{u-3}$, then by Observation~\ref{obs:four} and the inductive hypothesis, we would have that $w$ has a natural color of $u-2$, which contradicts the assumption that $w$ has a natural color of $u$. We also cannot have that $(i-1,j) \in \Gamma^k_{u-2}$ as then $w'$ would also be in $L'^k$ by the inductive hypothesis, which contradicts the assumption that $i$ and $j$ are even. Hence, we have $(i-1,j) \in \Gamma^k_{u-1}$. By applying 
Observation~\ref{obs:four} again, we then have $(i,j) \in \Gamma^k_u$.

The case of $j>0$ is analogous.

\textbf{Odd $u$, $\Leftarrow$:} 
Consider now an odd $u \geq 0$ and let $w$ be a word accepted by $C^k_u$ but rejected by $C^k_{u+1}$. We have that the unique greatest level $(i,j)$ with $w \in L^k_i \cap \hat L^k_j$ is such that $u \leq i+j \leq u+1$ with either odd $i$ or odd $j$ (or both).

We show that that the natural color of $w$ is $u$. To see this, consider a word $w'$ that extends $w$ and let $(i',j')$ be a unique greatest pair such that $w' \in L^k_{i'} \cap L^k_{j'}$. By Lemma~\ref{lem:movingLemma}, we have that $i' \leq i$ and $j' \leq j$. If either $i'$ or $j'$ stay odd, then $w' \notin L'^k$ and the condition for the word $w$ to have a natural color of $u$ is fulfilled. Alternatively, by Observation~\ref{obs:five}, $w'$ is rejected by $C^k_u$, which by the inductive hypothesis suffices to show that $w'$ has a natural color of at most $u-1$. Since thus for all extensions $w'$ of $w$, the
resulting word is still not in $L'^k$ or its resulting natural color is at most $u-1$, the claim follows.

\textbf{Odd $u$, $\Rightarrow$:} 
Let now $w$ be a word with a natural color of $u$ and $u$ be odd. The word is hence not in $L'^k$ and for the unique greatest $(i,j)$ such that $w \in L^k_i \cap \hat L^k_j$ we have that either $i$ is odd or $j$ is odd, as one of the COCOA whose language intersection is taken needs to reject $w$ for the word not to be in $L'^k$. 

We need to prove that $(i,j)$ is in $\Gamma^k_u$ and the word is hence in $\mathcal{L}(C^k_u)$. Since then $(i,j)$ is in $\Gamma^k_u$ and by Observation~\ref{obs:three}, we then have that $w$ is not in $\mathcal{L}(C^k_{u+1})$.

Consider the set of words $w'$ that result from $w$ by residual language invariant word injections. Ignore those words $w'$ in $L'^k$, as there is nothing to be proven for them. By assumption, all such words have a natural color of at most $u-1$. We build the sets of all index pairs $Y$ that can occur for any such word $w' \in L'^k$. Note that all such pairs are in $\Gamma_0 \cup \Gamma_2 \cup \ldots \Gamma_{u-1}$ (by the assumption that the natural color of $w$ is at most $u$). 
Let $i_\mathit{max} = \max\{ i' \in \NN \mid \exists j'. (i',j') \in Y\}$ and $j_\mathit{max} = \max\{ j' \in \NN \mid \exists i'. (i',j') \in Y\}$. We distinguish three cases:
\begin{itemize}
\item $i$ is odd and $j$ is even: Then, $j = j_\mathit{max}$ as $(i-1,j_\mathit{max}) \in Y$, and $i = 1 + i_\mathit{max}$, as any value $i > 1 + i_\mathit{max}$ would allow $w$ to be extended to a word contained in $L^k_{i_\mathit{max}+2} \cap \hat L^k_j$. This would contradict that $u$ is the natural color of $w$, as a word whose unique greatest pair is $(i_\mathit{max}+2,j)$ can be extended to a word in $(i_\mathit{max}+1,j)$ with a different containment in $L'^k$, and since that word, by the inductive hypothesis, cannot have color $u-1$, this means that $w$ cannot have a natural color of $u$.

Since $(i_\mathit{max},j)$ is in $\Gamma^k_{u-1}$ (as for $u$ to have a natural color of $u$, all extensions with different containment in $L'^k$ can have a color of at most $u-1$ and a lower natural color is not possible by $(i_\mathit{max},j)$ otherwise also having a natural color of less that $u-1$, which contradicts that $w$ has a natural color of $u$), by the definition of the sets $\{\Gamma^k_m\}_{0 \leq m \leq 2k}$, we have $(i_\mathit{max}+1,j) \in \Gamma^k_{u}$.

\item $i$ is even and $j$ is odd: This case is analogous to the previous case.
\item $i$ is odd and $j$ is odd. Then, $(i_\mathit{max}+1,j_\mathit{max}+1)$ is the unique greatest tuple such that $w'$ can be located at this level when the acceptance of $w'$ differs from the acceptance of $w$ (by Lemma~\ref{lem:movingLemma}). It cannot be located at a greater tuple by the same reasoning as in the first case.
By the assumption that $w$ has a natural color of $u$ and hence all extensions of $w$ with different containment in $L'^k$ need to have a natural color of at most $u-1$, we know that $(i_\mathit{max},j_\mathit{max}) \in \Gamma^k_{u-1} \cup \ldots \cup \Gamma^k_{1}$. At the same time, $(i_\mathit{max}+1,j_\mathit{max}+1) \notin \Gamma^k_{u-1} \cup \ldots \cup  \Gamma^k_{1}$ as otherwise by the inductive hypothesis, $w$ would have a natural color of less than $u$.
By the definition of the $\{\Gamma^k_m\}_{0 \leq m \leq 2k}$ sets, we then have $(i,j) = (i_\mathit{max}+1,j_\mathit{max}+1) \in \Gamma^k_u$. \qedhere
\end{itemize}
\end{proof}

Theorem~\ref{thm:conjunctionBecomesBig} provides a blueprint for building $\mathcal{C}^k$ from the COCOA for $L^k$ and $\hat L^k$. In particular, we can obtain $\mathcal{C}^k$ by a sequence of disjunction and conjunction operations. This characterization allows us to deduce that $\mathcal{C}^k$ must be of size exponential in $k$, as proposition~\ref{prop:exponentialSize} below shows. 
To simplify the construction an proof, we will perform the disjunction and conjunction operations on deterministic rather than history-deterministic co-Büchi automata and then only show that the resulting automata need to be exponentially big even when allowing history-determinism. The following lemma adapts folk results for taking the conjunction or disjunction of deterministic Büchi automata to the case of transition-based acceptance (and a co-Büchi acceptance condition).

\begin{lem}
\label{lem:conjunctionDisjunctionOfDCW}
Let $\mathcal{A}_1, \ldots, \mathcal{A}_n$ be deterministic co-Büchi automata over the same alphabet. We can construct a deterministic co-Büchi automaton $\mathcal{A}^\wedge$ for the conjunction of these languages of size $|\mathcal{A}_1| \cdot \ldots \cdot |\mathcal{A}_n|$, and a deterministic co-Büchi automaton $\mathcal{A}^\vee$ for the disjunction of theses languages of size $|\mathcal{A}_1| \cdot \ldots \cdot |\mathcal{A}_n| \cdot n$.
\end{lem}

\begin{proof}
The claim can be shown by adapting existing (folk) approaches to the case of transition-based acceptance. The resulting construction is, to the best of the author's knowledge, not found explicitly anywhere in the available literature. However, the \texttt{spot} framework for constructing and manipulating automata \cite{DBLP:conf/cav/Duret-LutzRCRAS22} has an implementation for pairs of automata.

Let for each $1 \leq i \leq n$ the automaton $\mathcal{A}_i$ be given as a tuple $(Q^i,\Sigma,\delta^i,q^i_0)$. We define $\mathcal{A}^\wedge = (Q^\wedge,\Sigma,\delta^\wedge,q_0^\wedge)$ with $Q^\wedge = Q^1 \times \ldots \times Q^n$, $q^\wedge_0 = (q^1_0,\ldots,q^n_0)$, and for each $(q_1, \ldots, q_n),(q'_1, \ldots, q'_n) \in Q^\wedge$, $x \in \Sigma$, and $c \in \{1,2\}$, we have $((q_1, \ldots, q_n),x,((q'_1, \ldots, q'_n),c')) \allowbreak{} \in \delta^\wedge $ 
if there exists some $c_1, \ldots, c_n, c'$ such that for every $1 \leq i \leq n$, we have $(q_i,x,q'_i,c_i) \in \delta^i$, and $c' = \max(c_1, \ldots, c_n)$.

Note that this automaton is deterministic and accepts a word if and only if it is accepted by every input co-Büchi automaton. This is because $\mathcal{A}^\wedge$ simulates all automata in parallel and has an accepting transition if and only if all transitions in the component automata are accepting. Hence, if and only if for some word $w$, eventually only accepting transitions are taken along a run for $w$ in $\mathcal{A}^\wedge$, then all automata accept.

We also define $\mathcal{A}^\vee = (Q^\vee,\Sigma,\delta^\vee,q_0^\vee)$ with $Q^\vee = Q^1 \times \ldots \times Q^n \times \{1, \ldots, n\}$, $q^\vee = (q^1_0,\ldots,q^n_0,1)$, and for each $(q_1, \ldots, q_n,j),(q'_1, \ldots, q'_n,j') \in Q^\vee$, $x \in \Sigma$, and $c \in \{1,2\}$, we have $((q_1, \ldots, q_n,j),x,((q'_1, \ldots, q'_n,\allowbreak{} j'),\allowbreak{} c')) \in \delta^\wedge $ 
if there exists some $c_1, \ldots, c_n, c'$ such that for every $1 \leq i \leq n$, we have $(q_i,x,q'_i,x_i) \in \delta^i$, $c' = 1$ if $j=n$ and $j'=1$ and $c'=2$ otherwise, and we have $j' = (j \mod n)+1$ if $c_j = 1$ and $j'=j$ otherwise.

This automaton is also deterministic. To see that it accepts the union of languages of $\mathcal{A}_1,  \ldots, \mathcal{A}_n$, consider the case that a word $w$ is accepted by some automaton $\mathcal{A}_i$. Then, we have that along the run of $w$ for $\mathcal{A}^\vee$, eventually the counter (i.e., the last element of the state set) gets stuck at value $i$, and no more rejecting transitions are taken afterwards. On the other hand, if all automata $\mathcal{A}_1, \ldots, \mathcal{A}_n$ reject the word $w$, then the counter moves through all possible values infinitely often along the run for $w$, and then because when the counter switches from $n$ to $1$, a rejecting transition is taken, the automaton $\mathcal{A}^\vee$ rejects the word.
\end{proof}

We use this lemma in the proof of the following proposition:

\begin{prop}
\label{prop:exponentialSize}
Let $\mathcal{C}^k = (C^k_1, \ldots, C^k_{2k})$ be a COCOA for $L^k \cap \hat L^k$. 

A minimal deterministic co-Büchi automaton for $\mathcal{L}(C^k_i)$ for some $1 \leq i \leq 2k$ has at most $2^{2k} \cdot k$ many states.
For even $k$, we have that the \hdcw{} $C^k_k$ has at least $2^{k}$ many states. For odd $k$, we have that $C^k_{k-1}$ has at least $2^{k-2}$ many states.
\end{prop}
\begin{proof}
Let for all languages $L^k_i$ and $\hat L^k_j$ be the respective two-state deterministic automata be denoted by $A^k_i$ and $\hat A^k_j$, which by Lemma~\ref{lem:gfgcoBuchiSizes} have two states each.

For the first part, note that for all $0 \leq i \leq 2k$, the set $\Gamma^k_u$ has at most $k$ many \emph{non-dominated elements}, i.e., pairs $(i,j)$ that do not have another different pair $(i',j')$ in the set such that $i'\geq i$ and $j' \geq j$. When building $C^k_u = \bigcup_{(i,j) \in \Gamma^k_u} A^k_i \cap \hat A^k_j$, only the non-dominated pairs have to be considered, as all words accepted by $A^k_{i'} \cap \hat A^k_{j'}$ for a dominated pair $(i',j')$ are also accepted by $A^k_{i} \cap \hat A^k_{j}$ for some non-dominated pair $(i,j)$.
A deterministic co-Büchi automaton $A^k_i \cap \hat A^k_j$ for some pair $(i,j)$ only needs 4 states by Lemma~\ref{lem:conjunctionDisjunctionOfDCW}. Taking the union of $k$ many such automata yields an automaton with at most $2^{2k} \cdot k$ many states by the same lemma.

For the second part, consider the case of
$C^k_k = \bigcup_{(i,j) \in \Gamma^k_k} A^k_i \cap \hat A^k_j$ for even $k$. Here, we take the disjunction of $\frac{k}{2}$ many 4 state automata, yielding an automaton with at most $2^k$ many states. This is at the same time also the lower bound, because the number of residual languages of $\mathcal{L}(C^k_k)$ is $2^k$. This is because every language $L^k_i \cap \hat L^k_j$ for $(i,j) \in \Gamma^k_u$ has the word suffix $\tilde w = (a_{4k-2i+1} a_{2j-2})^\omega$ that is not accepted by any $L^k_{i'} \cap \hat L^k_{j'}$ with $(i',j') \in \Gamma^k_u$ as well for $(i,j) \neq (i',j')$ and that is only in $L^k_i \cap \hat L^k_j$ for prefixes for which the letters $X_i$ and $Y_j$ each occur an even number of times in the prefix. This implies that a \hdcw{} for $C^k_k$ has different residual languages for words that differ in their numbers of $X_i$ or $Y_j$ letters for $(i,j) \in \Gamma^k_k$. The number of residual languages is hence $2^k$.
As minimal canonical \hdcw{} are \emph{semantically deterministic} \cite{DBLP:journals/lmcs/RadiK22}, having $2^k$ many residual languages implies a lower bound of $2^k$ for the size of $C^k_k$.
The case for $k$ being odd is analogous.
\end{proof}
Let us finally discuss that a similar blow-up does not occur when representing $L^k$ and $\hat L^k$ as deterministic parity automata.
\begin{prop}
\label{proposition:paritySizes}
Each of the languages $L^k$ (from Def.~\ref{def:toyLanguage}) and $\hat L^k$ (from Def.~\ref{def:toyLanguageB}) can be represented as deterministic parity automata with $2^k$ states (and not less).

There exists a deterministic parity automaton for $L^k \cap \hat L^k$ with no more than $2^{4k^2} \cdot k^{2k}$ many states.
\end{prop}
\begin{proof}
We can build a DPW $\mathcal{P}^k = (Q,\Sigma,\delta,q_0)$ for $L^k$ with $Q = \mathbb{B}^k$, $q_0 = (0, \ldots,0)$, and for all $(b_1,\ldots, \allowbreak{} b_k)\in Q$ and $x \in \Sigma$, we have $\delta((b_1,\ldots, b_k),x) = ((b'_1, \ldots, b'_k),c)$ for
$b'_i = \neg b_i$ if $x=X_i$ and $b'_i = b_i$ (for all $1 \leq i \leq k$) otherwise. The value of $c$ in this transition is defined as follows::
\begin{equation*}
c = \begin{cases} 0 &\text{ if } x \in \{X_1, \ldots, X_k, Y_1, \ldots, Y_k\} \\
i &\text{ if } x=a_{4k-2i} \text{ for some } 1 \leq i \leq k \\
i &\text{ if } x=a_{4k-2i+1} \text{ and } b_i=0 \text{ for some } 1 \leq i \leq k \\
i-1 &\text{ if } x=a_{4k-2i+1} \text{ and } b_i=1 \text{ for some } 1 \leq i \leq k \\
k & \text{ otherwise}. 
\end{cases}
\end{equation*}
We show that the automaton recognizes each word with its natural color (w.r.t.~$L^k$). For doing so, we 
compare the color induced by a run of $\mathcal{P}^k$ for some word $w$ against in which languages of $L^k_1, \ldots, L^k_k$ the word is.
\begin{itemize}
\item If there are infinitely many occurrences of letters in $\{X_1, \ldots, X_k, Y_1, \ldots, Y_k\}$, then the word is not in $L^k_1$. This word is recognized by $\mathcal{P}^k$ with color $0$, which is the natural color of the word (w.r.t.~$L^k$).
\item If there are infinitely many letters $a_{4k-2i}$ for some $1 \leq i \leq k$ but only finitely many letters $a_j$ for an index $j<4k-2i$, then the word is in $L^k_i$ but not in $L^k_{i+1}$. The word is recognized by $\mathcal{P}^k$ with color $i$ in this case.
\item If there are infinitely many letters $a_{4k-2i+1}$ for some $1 \leq i \leq k$ but only finitely many letters $a_j$ for an index $j<4k-2i+1$ and the number of $X_i$ letters is even, then the word is in $L^k_i$ but not in $L^k_{i+1}$. The word is recognized by $\mathcal{P}^k$ with color $i$ in this case.
\item If there are infinitely many letters $a_{4k-2i+1}$ for some $1 \leq i \leq k$ but only finitely many letters $a_j$ for an index $j<4k-2i+1$ and the number of $X_i$ letters is odd, then the word is in $L^k_{i-1}$ but not in $L^k_{i}$. The word is recognized with color $i-1$ by $\mathcal{P}^k$.
\item Otherwise, the word ultimately consists of letters $a_j$ with $j<4k-2k$. Then, even $L^k_k$ accepts the word, and the natural color of the word is  $k$. We have that $\mathcal{P}^k$ recognizes the word with color $k$.
\end{itemize}

Note that $\mathcal{P}^k$ is the smallest deterministic parity automaton for $L^k$ as it has $2^k$ many states and the number of residual languages of $L^k$ is $2^k$, so it cannot be smaller. %

A similar DPW can be built from $\hat L^k$ by replacing $X_i$ characters with $Y_i$ and renumbering the indices for the $a_j$ letters.

For the DPW for $L^k \cap \hat L^k$, we employ Proposition~\ref{prop:exponentialSize} to obtain deterministic co-Büchi automata of size at most $2^{2k} \cdot k$ for each of the levels of a COCOA for $L^k$ and then build a product parity automaton of the deterministic automata as in 
Proposition~\ref{proposition:COCOAtoDPWExponentialIsEnough},
which yields a deterministic parity automaton for $L'^k$ of size at most $(2^{2k} \cdot k)^{2k} = 2^{4k^2} \cdot k^{2k}$.
\end{proof}

Proposition~\ref{proposition:paritySizes}, Proposition~\ref{prop:exponentialSize}, and Lemma~\ref{lem:gfgcoBuchiSizes} together show that while $L^k$ and $\hat L^k$ can be represented with a COCOA that is exponentially more concise than any deterministic parity automaton for these languages, exponential conciseness is lost when computing a COCOA for $L^k \cap \hat L^k$.

\begin{rem}
\label{remark:deMorgan}
Exponential conciseness can also be lost when taking the disjunction of two COCOA (instead of taking their conjunction).
\end{rem}
\begin{proof}
COCOA for the complements of the languages of $L^k$ and $\hat L^k$ can be obtained by adding a \hdcw{} accepting the universal language as new first automaton in the chains, moving all chain elements one element back. After taking the conjunction of the resulting COCOA for the complement language, we obtain a result COCOA in which the first automaton accepts the universal language (by the construction in Theorem~\ref{thm:conjunctionBecomesBig}). Removing it yields a COCOA for $L^k \cup \hat L^k$. Applying Proposition~\ref{prop:exponentialSize} for the conjunction automaton yields the exponential lower size bound on the COCOA for $L^k \cup \hat L^k$. For the parity automata that we compare with, complementation can be performed without blow-up by adding $1$ to each transition color.
\end{proof}

\section{Complementing COCOA}
\label{sec:complementation}

In this section and as the final contribution of this paper, we consider the problem of complementing chains of co-Büchi automata. 
We show that such a complementation requires restructuring with which natural colors the words are recognized. It can happen that two words $w$ and $w'$ are both recognized by a COCOA with some color $i$, but a COCOA for the complement recognizes $w$ with some color $i-1$ while $w'$ is recognized with color $i+1$. As a consequence, complementing a COCOA is, unlike for deterministic parity automata, where complementing its language can be performed by adding one to each transition color, a non-trivial operation. 
Even more, complementation can lead to an exponential growth of the COCOA, as we show in this section of the paper. 

We define a family of COCOA $\mathcal{C}^1, \mathcal{C}^2, \ldots$ for which each COCOA $\mathcal{C}^k$ has $(4+3 \cdot k)$ many states overall, but the first automaton in a COCOA for the complement of $\mathcal{C}^k$ needs at least $2^k$ many states. Figure~\ref{fig:bluePrintComplementation} depicts the co-Büchi automata of a COCOA $\mathcal{C}^k$. 

There are three sets of letters in the alphabet $\Sigma$ of a language $L^k$ represented by the COCOA $\mathcal{C}^k$, namely $\{a_1, \ldots, a_{2k+1}\}$, $\{X_1, \ldots, X_k\}$, and $\{Y_1, \ldots, Y_k\}$.
Every automaton $A^k_{i}$ for $i \geq 2$ accepts those words for which letters from $X_1, \ldots, X_k$ appear only finitely often, letters from $\{a_{2k-i+5}, \ldots, a_{2k+1}\}$ appear only finitely often, and either $a_{2k-i+4}$ appears only finitely often or after the first $Y_{i-1}$ letter, the letter $X_{i-1}$ does not appear. 
Hence, if $Y_{i-1}$ and $X_{i-1}$ appear in a word (in that relative order), the set of letters from $a_{1}, \ldots, a_{2k+1}$ that may appear infinitely often along an accepted word changes.
The automaton $A_1$ accepts all words in which letters from  $X_1, \ldots, X_k$ appear only finitely often.

As in the previous section, we need to prove that for every $k \in \NN$, $\mathcal{C}^k$ is a valid COCOA, i.e., that it recognizes each word with its natural color w.r.t~the language represented by the chain of automata.

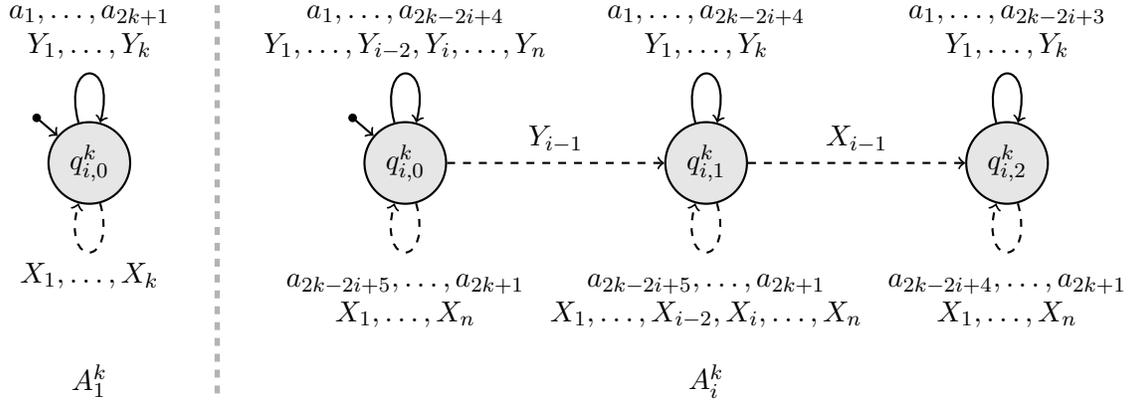
\begin{figure}
\centering\begin{tikzpicture}
\path[use as bounding box] (-2.3,-3.15) rectangle (12.6,2.2);

\begin{scope}[xshift=-1.2cm]
\node[state,thick,fill=black!10!white] (a) at (0,0) {$q^k_{i,0}$};
\draw[thick,->] (a) to[loop above] node[above] {\begin{tabular}{c}$a_1, \ldots, a_{2k+1}$\\ $Y_1, \ldots, Y_k$ \end{tabular}} (a);
\draw[thick,->,dashed] (a) to[loop below] node[below] {$X_1, \ldots, X_k$} (a);

\draw[fill=black] ($(a)+(-0.7,0.6)$) circle (0.05cm);
\draw[->,thick] ($(a)+(-0.7,0.6)$) -- (a);

\end{scope}

\node[state,thick,fill=black!10!white] (b) at (3,0) {$q^k_{i,0}$};
\node[state,thick,fill=black!10!white] (c) at (7,0) {$q^k_{i,1}$};
\node[state,thick,fill=black!10!white] (d) at (11,0) {$q^k_{i,2}$};

\draw[thick,->] (b) to[loop above] node[above] {\begin{tabular}{c}$a_1, \ldots, a_{2k-2i+4}$\\$Y_1, \ldots, Y_{i-2}, Y_{i}, \ldots, Y_n$\end{tabular}} (b);

\draw[thick,->] (c) to[loop above] node[above] {\begin{tabular}{c}$a_1, \ldots, a_{2k-2i+4}$\\$Y_1, \ldots, Y_k$\end{tabular}} (c);

\draw[thick,->,dashed] (b) to[loop below] node[below] {\begin{tabular}{c}$a_{2k-2i+5}, \ldots, a_{2k+1}$\\$X_1, \ldots, X_n$ \end{tabular}} (b);

\draw[thick,->,dashed] (c) to[loop below] node[below] {\begin{tabular}{c}$a_{2k-2i+5}, \ldots, a_{2k+1}$\\$X_1, \ldots, X_{i-2}, X_i, \ldots, X_n$ \end{tabular}} (c);

\draw[thick,->] (d) to[loop above] node[above] {\begin{tabular}{c}$a_1, \ldots, a_{2k-2i+3}$\\$Y_1, \ldots, Y_k$\end{tabular}} (d);

\draw[thick,->,dashed] (d) to[loop below] node[below] {\begin{tabular}{c}$a_{2k-2i+4}, \ldots, a_{2k+1}$\\$X_1, \ldots, X_n$ \end{tabular}} (d);

\draw[thick,->,dashed] (b) to node[above] {$Y_{i-1}$} (c);
\draw[thick,->,dashed] (c) to node[above] {$X_{i-1}$} (d);

\draw[fill=black] ($(b)+(-0.7,0.6)$) circle (0.05cm);
\draw[->,thick] ($(b)+(-0.7,0.6)$) -- (b);

\node at ($(a)+(0,-2.9)$) {$A^k_1$};
\node at ($0.5*(b)+0.5*(d)+(-0.0,-2.9)$) {$A^k_i$};

\draw[line width=2pt,dashed,color=black!30!white] (0.5,2.1) -- +(0,-5.25);

\end{tikzpicture}
\caption{Co-Büchi automata $A^k_1$ and $A^k_i$ (for $2 \leq i \leq k+1$) for the COCOA $\mathcal{C}^k = (A^k_1, \ldots, A^k_{k+1})$ (for some $k \in \NN$) used for showing that complementation can cause an exponential blow-up in Section~\ref{sec:complementation}}
\label{fig:bluePrintComplementation}
\end{figure}

\begin{lem}
Let $\mathcal{C}^k = (A^k_1, \ldots, A^k_{k+1})$ be a sequence of automata as given in Figure~\ref{fig:bluePrintComplementation} (for some $k \in \NN$), and $L^k = \mathcal{L}(A^k_1, \ldots, A^k_{k+1})$.

We have that $\mathcal{C}^k$ is a valid COCOA, i.e., each automaton $A^k_i$ for $1 \leq i \leq k+1$ accepts exactly the words with a natural color of at least $i$ w.r.t.~$L^k$.
\end{lem}

\begin{proof}
First of all note that all automata in the chain are history-deterministic as they are also deterministic.
Also note that for every $1 \leq i \leq k$, we have $\mathcal{L}(A^k_i) \supset \mathcal{L}(A^k_{i+1})$ as all automata in the chain only accept words that ultimately end in letters from $\{Y_1, \ldots, Y_k, a_1, \ldots, a_{2k+1}\}$, but the letters from $a_{1}, \ldots, a_{2k+1}$ that may appear infinitely often strictly shrinks from level to level.

We now prove by induction that for each $1 \leq i \leq k$, the language $L^k_i$ contains exactly those words whose natural color regarding $L^k$ is at least $i$.

\emph{Induction basis:} Let $w$ be a word. We need to show that exactly the words $w \in L^k$ for which every sequence of residual-language invariant word injections leading to a word $w'$ that has $w' \in L^k$ as well are the ones \emph{not} in the language of $A^k_1$. 

$\Rightarrow$: So let $w$ be a word that is not in the language of $A^k_1$. 
This means that some letter in $X_1, \ldots, X_k$ appears infinitely often in the word. Injecting additional letters does not change this, so for the resulting word, we have that it is in $L^k$ and not in the language of $A^k_1$ as well.

$\Leftarrow$: Let $w \in L^{k}$ be a word for which for some $J$, every sequence of residual-language invariant word injections leads to a word that is also in $L^k$. Then this includes the word resulting from injecting $a_{2k+1}$ infinitely often, which does not change whether the word is accepted by $A^k_1$ but ensures that it is rejected by $A^k_2, \ldots, A^k_{k+1}$. Also, this injection is residual-language invariant w.r.t.~$L^k$ as all co-Büchi automata in $\mathcal{C}^k$  always self-loop under this letter. Hence, if $w$ was in $L^k$ before the injections, it was already rejected by $A^k_1$.

\emph{Induction step:} 
$\Rightarrow$: Let $w$ be a word not in $\mathcal{L}(A^k_{i})$. We need to show that $w$ has a natural color of at most $i-1$. 
If the word is not in $\mathcal{L}(A^k_{i-1})$, then this holds by the inductive hypothesis. So it remains to consider the case that $A^k_{i-1}$ accepts the word. 

Then, we have that either 
the word contains the $a_{2k-2i+5}$ letter infinitely often or it contains the $X_{i-1}$ letter after the first $Y_{i-1}$ letter and $a_{2k-2i+4}$ infinitely often. Injecting any residual-language invariant words cannot change that. In this context, note that if the word contains a $Y_{i-1}$ letter and not a $X_{i-1}$ letter afterwards, then any injection of $X_{i-1}$ letters after the first $Y_{i-1}$ letter is not suffix-language invariant as it changes whether $(a_{2k-2i+4})^\omega$ is in the residual language. If we now inject a word containing letters from $a_{2k-2i+6}, \ldots, a_{2k+1}$ infinitely often, then the word becomes rejected by $A^k_{i-1}$, proving by the inductive hypothesis that the resulting word has a natural color of at most $i-2$. Otherwise, the word stays accepted by $A^k_{i-1}$ and rejected by $A^k_{i}$, retaining whether the word is in $L$. As for all residual-language invariant injections, we now know that either the resulting word is in $L$ if and only if it is in $L$ before the injection or the resulting word has a natural color of at most $i-2$, this direction of the induction step is complete.

$\Leftarrow$: Let $w$ be a word with a natural color of at most $i-1$. We need to prove that we have that $w$ is not in $\mathcal{L}(A^k_{i})$. 
If the word has a color of at most $i-2$, then we know that it is not in $\mathcal{L}(A^k_{i-1})$ by the inductive hypothesis, so by the strict language inclusion between the chain elements, the claim already holds. So assume that the natural color of $w$ is exactly $i-1$.
Every suffix-language invariant word injection changing whether the word is in $L^k$ hence needs to lead to a word with a color of at most $i-2$. 

Assume for a proof by contradiction that $w$ is accepted by $A^k_i$.
Now inject $a_{2k-2i+3}$ infinitely often into $w$. The word is still in $\mathcal{L}(A^k_{i})$ but is rejected by $\mathcal{L}(A^k_{i+1})$, so it is in $L$ if and only if $i$ is even. 
Since it does not have a natural color of at most $i-2$ but the resulting word can also not have a natural color of $i-1$ because the evenness of $i-1$ does not fit to whether the word is contained in $L$, this proves that the resulting word has a natural color of at least $i$, which cannot happen if $w$ has a natural color of exactly $i-1$.
\end{proof}

Let us now consider the problem of complementing a COCOA $\mathcal{C}^k$. The following theorem captures that the first element of a COCOA for the complement of $\mathcal{C}^k$ needs a number of states that is exponential in $k$.

\begin{thm}
Let, for some $k \in \NN$, $\mathcal{C}^k = (A^k_1, \ldots, A^k_{k+1})$ be a chain of co-Büchi automata with the structure given in Figure~\ref{fig:bluePrintComplementation}, and $L^k = \mathcal{L}(A^k_1, \ldots, A^k_{k+1})$.

Every first chain element of a COCOA encoding $\Sigma^\omega \setminus L^k$ needs at least $2^k$ many states.
\end{thm}
\begin{proof}
We perform the proof by characterizing the words with a natural color of exactly $0$. We denote the language formed by these words as $\hat L^k_0$ henceforth and prove that it has at least $2^k$ many residual languages. As the first element of a COCOA for $\Sigma^\omega \setminus L^k$ accepts exactly the complement of $\hat L^k_0$, and minimized co-Büchi automata are, w.l.o.g., language-deterministic (as the minimization algorithm by Abu Radi and Kupferman \cite{DBLP:journals/lmcs/RadiK22} produces such automata), this proves that the first element of a COCOA for $\Sigma^\omega \setminus L^k$ needs at least $2^k$ many states.
We can characterize $\hat L^k_0$ as follows:
\begin{align*}
\hat L^k_0 & = \{w_0 w_1 \ldots \in \Sigma^\omega \mid \forall 1 \leq i \leq k. \exists j \in \NN. Y_i \notin \{w_0, \ldots, w_{j-1}\} \wedge w_j = Y_i \wedge X_i \notin \\ & \quad \quad \{w_{j+1}, w_{j+2}, \ldots \},
\exists_\infty j \in \NN. w_j \in \{ a_{2k}, a_{2k+1} \} \}
\end{align*}
In this equation, the expression $\exists_\infty j \in \NN$ denotes that there exist infinitely many elements $j \in \NN$ with the stated properties.
To prove that $L^k_0$ contains exactly the words of $\Sigma^\omega \setminus L^k$ with a natural color of $0$, consider a word in $\hat L^k_0$. It contains all letters in $Y_1, \ldots, Y_k$ but no letter in $X_0, \ldots, X_n$ after the corresponding $Y_i$ letter. 
Hence, $A^k_1$ accepts the word, but $A^k_2$ rejects it, making the word contained in $\Sigma^\omega \setminus L^k$.

Now consider a sequence of injection points $J$ that are all placed after each letter from $\{Y_1, \ldots, Y_n\}$ has occurred at least once.
Any run for any of the automata $A^k_i$ for $2 \leq i \leq k+1$ will be in the state $q^k_{i,1}$ after all $Y_i$ letters have been seen. Any successive letter $X_i$ (for $1 \leq i \leq k$) will change the residual language of the COCOA (as it changes whether the residual word $(a_{2k+2i+4})^\omega$ is accepted), so residual-language invariant word injections at positions in $J$ cannot contain $X_i$ letters. Without injecting any $X_i$ letter, the resulting word will still be accepted by $A^k_1$ and still be rejected by $A^k_2$, making the word contained in $\Sigma^\omega \setminus L^k$ again.

For the converse direction, consider a word $w$ in $\Sigma^\omega \setminus L^k$ for which a sequence $J$ exists all of whose residual-language invariant word injections are also in $\Sigma^\omega \setminus L^k$. We distinguish three cases:
\begin{itemize}
\item Some letter $Y_{i-1}$ may not be contained in $w$. Then $q_{i,1}$ is not reached along the run of $A^k_i$ for $w$, and since injecting $X_{i-1}$ makes all automaton runs for the automata in the COCOA self-loop, we have that this injection is residual-language invariant. Injecting $X_{i-1}$ infinitely often however leads to the resulting word having a natural color of $0$ w.r.t.~$L^k$, making it not contained in $\Sigma^\omega \setminus L^k$, which is a contradiction.
\item Alternatively, some letter $X_{i-1}$ appears after the first occurrence of $Y_{i-1}$ in $w$. But then, the (only) run of $A^k_i$ for $w$ would move to state $q^k_{i,2}$ after the occurrence of letter $X_{i-1}$, making injections of $X_{i-1}$ after that point residual-language invariant. By the same reasoning as in the last case, we get a contradiction.
\item Finally, we may have that for all $1 \leq i \leq k$, the letter $X_{i}$ does not occur after the last $Y_{i}$ letter in $w$, and these $Y_{i}$ letters can be found somewhere in the word for all $1 \leq i \leq k$. In this case, every resulting word has a natural color of at least $1$ in $\mathcal{C}^k$ (as $A^k_1$ accepts it). If no such injection changes that the word has an even natural color w.r.t. $L^k$, then we have that $A^k_2$ accepts it and $A^k_3$ rejects it as otherwise by injecting $a_{2k-2}$ infinitely often we could make the resulting word not be contained in $L^k$ or the word is actually in $L^k$ (and hence not in $\Sigma^\omega \setminus L^k$). Hence, either $a_{2k}$ or $a_{2k+1}$ already occur infinitely often in $w$ (as the word is rejected by $A^k_3$), making it contained in $\hat L^k_0$.
\end{itemize} 
We now have that $\hat L^k_0$ has at least $2^k$ many residual languages: all words using only letters from $\{Y_1, \ldots, Y_k, a_1, \ldots, a_{2k+1}\}$ that differ by which letters from $\{Y_1, \ldots, Y_k\}$ have appeared so far differ by their residual language, in particular by which elements of $(\{Y_1, \ldots, Y_k\})^* \{a_1, \ldots, a_{2k+1}\}^\omega$ are in their residual languages.

As a minimal co-Büchi automata for $\Sigma^\omega \setminus \hat L^k_0$, which is the first co-Büchi automaton in a minimized COCOA for $\Sigma^\omega \setminus L^k$, can w.l.o.g.~be assumed to track the residual language,
we have that the first COCOA automaton for $\Sigma^\omega \setminus L^k$ needs to have at least $2^k$ many states.
\end{proof}

We note that just like for the conjunction and disjunction operations, the proof of the theorem stating an exponential blow-up for COCOA complementation employs the property of COCOA that they can have a number of residual languages that is exponential in the overall number of their states.

\section{Conclusion}
\label{sec:conclusion}

In this paper, we took a close look at the conciseness of chains of history-deterministic co-Büchi automata with transition-based acceptance (COCOA) %
by aggregating previous results on them, deriving corollaries from these previous results, and augmenting the resulting overview of known conciseness properties of COCOA with three new technical results that are not corollaries of existing work. 

In particular, we showed that by splitting a language to be represented into levels, COCOA can be exponentially more concise than deterministic parity automata even when the automata on the individual levels are not. Secondly, we showed that exponential conciseness can be shattered by performing language disjunction or conjunction operations.
Finally, we showed that even complementing a single COCOA may require an exponential blow-up. 
While the first of these results even holds for a language that only has a single residual language, the loss of conciseness in the second and third results was caused by some co-Büchi automaton in the resulting COCOA needing to represent an exponential number of residual languages.
None of the latter two technical results depend on the conciseness of history-deterministic co-Büchi automata over deterministic co-Büchi automata.

Apart from providing some insight into the capabilities and limits of COCOA as a representation for $\omega$-regular languages, our results inform future work on algorithms for performing operations on COCOA as well as future work on using COCOA for practical applications. In particular, the lower bounds on conjunction/disjunction/complementation operations provides a baseline that future algorithms for computing Boolean combinations of COCOA can be compared against. Furthermore, since in all results on Boolean combinations of COCOA the exponential size increase was caused by the fact that COCOA can have an exponential number of residual languages, our results motivate looking into alternative representations for $\omega$-regular languages that inherit the good properties of COCOA (such as polynomial-time minimization) but cannot represent a language with a number of residual languages that is exponential in the number of states. Such a model may then support Boolean operations with only a polynomial representation blow-up and potentially also permits simpler algorithms for performing such operations.

\bibliographystyle{alphaurl}
\bibliography{bib}

\end{document}